\documentclass[11pt]{article}
\usepackage{color}
\usepackage{graphicx}
\usepackage{amsmath}
\usepackage{amssymb}
\usepackage{mathrsfs}
\usepackage[numbers,sort&compress]{natbib}
\usepackage{amsthm}
\numberwithin{equation}{section}
\usepackage{pgf}
\usepackage{CJK}
\usepackage{graphicx}
\usepackage{tikz}
\usepackage{stmaryrd}
\usepackage{graphicx}
\usepackage[unicode]{hyperref}
\usepackage[latin1]{inputenc}
\usepackage[T1]{fontenc}
\usepackage[english]{babel}
\usepackage{listings}
\usepackage{xcolor,mathrsfs,url}
\usepackage{amssymb}
\usepackage{amsmath}
\usepackage{ifthen}
\newtheorem{theorem}{Theorem}
\newtheorem{lemma}{Lemma}
\newtheorem{remark}{Remark}

\newtheorem{Proposition}{Proposition}
\newtheorem{RHP}{RHP}

\usepackage {caption}
\usepackage{float}
\usepackage{subfigure}
\hypersetup{hypertex=true,
	colorlinks=true,
	linkcolor=blue,
	anchorcolor=blue,
	citecolor=red}

\DeclareMathOperator*{\res}{Res}

\begin{document}

	\baselineskip=15pt

	\title{\Large A Riemann-Hilbert approach   to the  two-component modified Camassa-Holm equation }
	\author{\Large Kai Xu$^1$, \  Luman Ju$^1$ and  \  Engui Fan$^{1}$}
	\footnotetext[1]{ \  School of Mathematical Sciences  and Key Laboratory for Nonlinear Science, Fudan University, Shanghai 200433, P.R. China.}
	\footnotetext[2]{ \   Email address:\ Kai XU: \ 22110180047@fudan.edu.cn; Luman JU: 21110180009@m.fudan.edu.cn; Engui FAN:  faneg@fudan.edu.cn}
	\date{ }
	\maketitle
	
	\begin{abstract}
		\baselineskip=16pt
In this paper,  we develop   a Riemann-Hilbert (RH)  approach to the Cauchy problem for the two-component modified Camassa-Holm (2-mCH) equation based on its Lax pair.
Further via  a series of deformations to the  RH  problem by using the $\bar{\partial}$-generalization of Deift-Zhou steepest descent method, we obtain  the long-time asymptotic approximations to the solutions   of the 2-mCH equation in four kinds of  space-time regions.  Especially we introduce a technique to unify
multi-jump matrix factorizations  into one form   which can greatly  simplify
the calculation  of the $\bar{\partial}$-steepest descent method.
\\[6pt]
    \end{abstract}

\tableofcontents

\section{Introduction}
\qquad In this paper, we develop a  Riemann-Hilbert (RH) approach  to  the Cauchy problem of the two-component modified Camassa-Holm (2-mCH) equation
\cite{JCZ}
\begin{align}
&(u-u_{xx})_t+[(u-u_x)(v+v_x)(u-u_{xx})]_x=0,  \label{2mch1}\\
&(v-v_{xx})_t+[(u-u_x)(v+v_x)(v-v_{xx})]_x=0,  \label{2mch2}\\
& u(x,0)=u_0(x), \ \   v(x,0)=v_0(x),\label{2mch3}
\end{align}
with the nonzero boundary conditions
\begin{equation}
u_0(x)\to 1,\quad v_0(x)\to1, \ \ x\to \pm \infty. \label{2mch4}
\end{equation}

The 2-mCH equation \eqref{2mch1}-\eqref{2mch2}   was proposed in  \cite{JCZ}
as a two-component integrable extension of the classical  mCH equation
\begin{align}
	&(u-u_{xx})_{t}+\left[(u-u_{xx})\left(u^{2}-u_{x}^{2}\right)\right]_{x} =0,  \label{mch}
\end{align}	
which  was first  presented   by Fokas \cite{Foks}  and Fuchssteiner  using  recursion operators \cite{Fuch}, and later  found  by Olver and
Rosenau \cite{PP}   via   tri-Hamiltonian duality to the bi-Hamiltonian of the mKdV
equation. Because of its famous mathematical structure and properties, the mCH equation (\ref{mch})  has gained lots of attention.
For example, Olver et.al found the  bi-Hamiltonian structure of the mCH equation \cite{PP,mch1}.
  Kang et.al presented the Liouville-type transformation  relating the isospectral problems for the mCH equation,
  and introduced a Miura-type map from the mCH equation to the CH equation \cite{mch0}.
  Boutet de Monvel et.al studied the long-time asymptotics for the CH equation  and the Degasperis-Procesi  equation \cite{BM1,BM2,BM3}, and then developed a RH approach and obtained  the long-time asymptotics for the Cauchy problem of the mCH equation \cite{mch7,BM4}.  
Yang and Fan discussed the long-time asymptotic behavior for the
Cauchy problem of the   mCH   equations  in
the solitonic regions    \cite{mch8}.
More works on the the CH,   mCH and DP  equations  can be found in   \cite{CH1,CH2,CH3,CH4,CH5,CH6,CH7,CH8,CH9,CH10,CH11,CH12,CH13,CH14,CH15,CH16,CH17,CH18,mch2,mch3,mch4,mch5,mch6}.

In recent years, the 2-mCH equation \eqref{2mch1}-\eqref{2mch2} also draws some attentions. For example,
the peakon   and multi-peakon solutions were constructed in  \cite{xia,2mch2};
   Tian and Liu obtained the 2-mCH equation and found its bi-Hamiltonian structure by applying tri-Hamiltonian duality into the bi-Hamiltonian representation of the Wadati-Konno-Ichikawa equation    \cite{KQP};  Chang et al presented an explicit construction of peakon solutions for the 2-mCH equation \eqref{2mch1}-\eqref{2mch2}, and discussed the global existence and the large-time asymptotics of the peakon solutions   \cite{CHS}. Hay et al found the bi-Hamiltonian structure of the 2-mCH system and proved that equation \eqref{2mch1}-\eqref{2mch2} possess infinite conservation laws \cite{hay}.
Zhang and Qiao  studied the periodic Cauchy problem for the 2-mCH equation \eqref{2mch1}-\eqref{2mch2}  \cite{qiao2020}.

 To the best of our knowledge,  the inverse scattering transform or Riemann-Hilbert 
 (RH) method  have not yet been applied to the  2-mCH  equation  \eqref{2mch1}-\eqref{2mch2}.
The purpose of our  paper is to develop   a   RH  method  to  the Cauchy problem   (\ref{2mch1})-(\ref{2mch4})
   and further analyze the  long-time asymptotic behavior 
 of the solution for the  2-mCH  equation.

We define a transformation
\begin{equation}
p(x,t):=u(x+t,t)-u_x(x+t,t),\quad
q(x,t):=v(x+t,t)+v_x(x+t,t),
\end{equation}
then  the Cauchy problem  (\ref{2mch1})-(\ref{2mch4})
is changed into the following Cauchy problem that we consider in our paper
\begin{eqnarray}
&&m_t +(wm)_x=0,\ \ n_t +(wn)_x=0, \ \ t>0,x\in\mathbb{R}\label{2mch-p}\\
&&  m=p+p_x,\ \ n=q-q_x,\ \ w=pq-1,\label{2mch-q}\\
&&p(x,0)=p_0(x),\ q(x,0)=q_0(x),\label{2mch-o}
\end{eqnarray}
with $p_0,q_0\in 1+H^{3,2}(\mathbb{R})$ satisfy  the following boundary  condition
  \begin{equation}\label{bdy}
\ p_0(x)\to 1,\ q_0(x)\to 1,\quad x\to \pm \infty,
\end{equation}
where the weighted Sobolev space $H^{3,2}(\mathbb{R})$ is defined by
\begin{equation}
H^{3,2}(\mathbb{R}):=\{u(x) \in L^2(\mathbb{R}): \langle \cdot \rangle^2 \partial^j_xu\in L^2(\mathbb{R}),\   j=0,1,2,3\}. \nonumber
\end{equation}

Our aim is to prove the following  results.

\begin{theorem}\label{last}   Let $p(x,t)$ and $q(x,t)$  be the solutions  for  the Cauchy  problem (\ref{2mch-p})-(\ref{bdy}) associated with generic data
  $p_0, q_0 \in 1+H^{3,2}(\mathbb{R})$  such that $m(x,t), n(x,t)>0 $.
Then  as  for all $ t\to\infty$,  we obtain the following asymptotic expansions

{\rm I.}  In  the region ${\xi}\in(-\infty,-1/4)\cup(2,+\infty)$ with     $  {\xi}=x/t$,
\begin{align}
&( \log p )_x  = p^{sol}(x,t) +\mathcal{O}(t^{-1+\kappa}),\nonumber\\
&( \log q )_x  =q^{sol}(x,t) +\mathcal{O}(t^{-1+\kappa}),\nonumber
\end{align}
where $0<\kappa<1/2$, $p^{sol}(x,t)$ and $q^{sol}(x,t)$ are functions associated with solitons given by \eqref{pr1}-\eqref{qr1}, and
$$ x(y,t)=y +S_1(y,t)+\mathcal{O}(t^{-1+\kappa}).
$$
with $S_1(y,t)$ being given by \eqref{s1},  and  $y$ being  a new spatial  variable given by 
the reciprocal transformation    \eqref{2.39}.  

{\rm II.}   In  the region ${\xi}\in(-1/4,0)\cup (0, 2)$,
\begin{align}
&( \log p )_x  =   p^{sol}(x,t) +g_1(x,t)t^{-1/2}+\mathcal{O}(t^{-1+\kappa}),\nonumber\\
&( \log q )_x  =  q^{sol}(x,t) +g_1(x,t)t^{-1/2}+\mathcal{O}(t^{-1+\kappa}),\nonumber
\end{align}
where $0<\kappa<1/2$, $p^{sol}(x,t)$ and $q^{sol}(x,t)$ are functions associated with solitons given by \eqref{pqr2}, $g_1(x,t)$ and $g_2(x,t)$ are
the interaction of the discrete spectrum with the continuous spectrum given by \eqref{g1}-\eqref{g2} and
$$
x(y,t)=y +S_2(y,t)+\mathcal{O}(t^{-1+\kappa}),$$
with $S_2(y,t)$ being given by \eqref{s2}.

\end{theorem}

The key tool to  prove  Theorem \ref{last} is the $\bar{\partial}$-generalization \cite{McL1,McL2} of the  Deift-Zhou steepest descent method \cite{DZ1,DZ2}.
In recent years, this method
	has been   successfully  used    to investigate  long-time asymptotics,  soliton
	resolution and  asymptotic stability of N-soliton solutions to
	integrable systems in a weighted Sobolev space
	\cite{DandMNLS,fNLS,Liu3,SandRNLS,YF1,YF3,YF4}.

\begin{remark}
    Compared with  the RH method to the   mCH  equation  (\ref{mch})    \cite{mch7, mch8}, it is found  that  extending the
      RH method  to the 2-mCH  equation \eqref{2mch-p}-\eqref{2mch-q} is not trivial from
the following three different  prominent features.

\begin{itemize}
\item[ $\blacktriangleright$]  In  the Lax pair (\ref{lax1})  of the  2-mCH equation,    the  off-diagonal elements for the matrices   $U$ and $V$
 are not symmetric, we   construct a  matrix transformation  (\ref{symeg})   to change the Lax pair (\ref{lax1})  into a new  one (\ref{lax2}).

\item[ $\blacktriangleright$]   The reconstruction formula  for  the   mCH  equation  (\ref{mch})
was   derived   from a reciprocal transformation   in \cite{mch7}.   However, the independent  reconstruction formulas of
 $m$ and $n$  for  the 2-mCH  equation \eqref{2mch-p}-\eqref{2mch-q} cannot be obtained from
our  reciprocal transformation (\ref{2.39}).  To overcome this difficulty,  we  apply the asymptotic expansion at $k=i$
 of the solution  $M(k)$ for the  RHP \ref{rhp1} to obtain  a suitable reconstruction  formula   (\ref{restr}).

\item[ $\blacktriangleright$]  In addition, we introduce a technique to unify
multi-jump matrix factorizations  into one form (\ref{v1}) and (\ref{v2})  which can greatly  simplify
the calculation  of the $\bar{\partial}$-steepest descent method. Especially   to  the case with   8 phase points  on $\mathbb{R}$ in Section \ref{sec4}, if by
 the way  used in  \cite{mch8},   there will be   32     opened  jump lines,
     32   regions  and  32  extension functions     for
  the 2-mCH equation \eqref{2mch-p}-\eqref{2mch-q},     we   cut  down them  from  the number 32   to  the   number 2.

\end{itemize}

\end{remark}

This paper is organized as follows.  In section \ref{sec2}, we  carry out direct scattering transform based on   the Lax pair of the 2-mCH equation,
and then  establish  a  RH  problem related to the  Cauchy problem \eqref{2mch-p}-\eqref{2mch-o}.
By  a series of deformations to the RH problem with  the $\bar{\partial}$-steepest descent method in Section \ref{sec3} and Section \ref{sec4}, respectively,  we
  then obtain   long-time asymptotic behavior for the solution of the 2-mCH equation in two regions without phase points and  two regions with  phase points.

\section{Inverse scattering transform } \label{sec2}

In this section,   we   study  the direct scattering transform and the  RH problem associated with the Cauchy problem
 (\ref{2mch-p})-(\ref{2mch-o}).

\subsection{Spectral analysis on Lax pair}

\quad
 The 2-mCH equation \eqref{2mch-p}-\eqref{2mch-q}  admits  the following Lax pair
\begin{equation}\label{lax1}
\Psi_x=U\Psi,\quad \Psi_t=V\Psi,
\end{equation}
where
\begin{eqnarray}
&&U=\frac{1}{2}\left(\begin{array}{cc}
-1&\lambda m\\-\lambda n&1\end{array}\right),\\
&&V=\left(\begin{array}{cc}
\lambda^{-2}+\frac{w}{2}&-\lambda^{-1}p-\frac{\lambda}{2}wm\\
\lambda^{-1}q+\frac{\lambda}{2}wn&-\lambda^{-2}-\frac{w}{2}
\end{array}\right).\qquad
\end{eqnarray}

To keep the  symmetry and zero trace of  the matrix $U$ in Lax pair (\ref{lax1}),   we first take the transformation
\begin{equation}
\hat{\Psi}(x,t,\lambda)=A(x,t)\Psi(x,t,\lambda) \label{symeg}
\end{equation}
with
\begin{equation}
A(x,t)=(mn)^{-\frac{1}{4}}\left(\begin{array}{cc}
\sqrt{m}&0\\0&\sqrt{n}\end{array}\right).
\end{equation}
Then the Lax pair \eqref{lax1} changes into
\begin{equation}\label{lax2}
\hat{\Psi}_x=\hat{U}\hat{\Psi},\quad \hat{\Psi}_t=\hat{V}\hat{\Psi},
\end{equation}
where
\begin{eqnarray}
&&\hat{U}=\frac{1}{2}\left(\begin{array}{cc}-1&\lambda\sqrt{mn}\\
-\lambda\sqrt{mn}&1\end{array}\right)+\left(\begin{array}{cc}\frac{mn_x-m_xn}{4mn}&0\\
0&-\frac{mn_x-m_xn}{4mn}\end{array}\right),\\
&&\hat{V}=\left(\begin{array}{cc}\lambda^{-2}+\frac{w}{2}&-\sqrt{\frac{n}{m}}[\lambda^{-2}p+\frac{\lambda}{2}wm]\\
\sqrt{\frac{m}{n}}[\lambda^{-1}q+\frac{\lambda}{2}wn]&-\lambda^{-2}-\frac{w}{2}\end{array}\right)
+\left(\begin{array}{cc}
\frac{mn_t-m_tn}{4mn}&0\\0&-\frac{mn_t-m_tn}{4mn}\end{array}\right).\qquad\quad
\end{eqnarray}
Notice that the coefficients of the Lax pair \eqref{lax2} have singularities at $\lambda=0$ and $\lambda =\infty$. To have a good control on the asymptotic
behavior of eigenfunctions, we use two different transformations to get new forms of \eqref{lax2}  respectively.
\subsubsection{Jost solutions at $\lambda=\infty$}

\quad Make the transformation
\begin{equation}
\hat{\Phi}(x,t,\lambda)=D(\lambda)\hat{\Psi}(x,t,\lambda)
\end{equation}
 with
\begin{equation}
D(\lambda)=\left(\begin{array}{cc}
1&-\frac{ \lambda}{1+\sqrt{1- \lambda^2}}\\
-\frac{ \lambda}{1+\sqrt{1- \lambda^2}}&1
\end{array}\right), \ \  \lambda=\frac{1}{2 }(k+k^{-1}),
\end{equation}
then $\hat{\Phi}$ satisfies the following Lax pair
\begin{equation}\label{lax2-1}
\left\{\begin{array}{lr}
\hat{\Phi}_x+\frac{i(k^2-1)}{4 k}\sqrt{mn}\sigma_3\hat{\Phi}=\tilde{U}\hat{\Phi},\\
\hat{\Phi}_t-\left[\frac{i(k^2-1)}{4 k}w\sqrt{mn}+\frac{2i k(k^2-1)}{(k^2+1)^2}\right]\sigma_3\hat{\Phi}=\tilde{V}\hat{\Phi},
\end{array}
\right.
\end{equation}
where
\begin{eqnarray}
&&\tilde{U}(k)=-\frac{ik}{k^2-1}\left( \sqrt{mn}-1  +\frac{mn_x-m_xn}{2mn}\right)\sigma_3-\frac{i(k^2+1)}{2(k^2-1)}\left(\sqrt{mn}-1+\frac{mn_x-m_xn}{2mn}\right)
\left(\begin{array}{cc}0&1\\-1&0\end{array}\right),\nonumber\\
&&\tilde{V}(k)=\frac{ik}{k^2-1}\left( w(\sqrt{mn}-1)-\frac{mn_t-m_tn}{2mn}\right)\sigma_3+\frac{i k}{k^2-1}\left(
\sqrt{\frac{n}{m}}p+\sqrt{\frac{m}{n}}q-2 \right)\sigma_3\nonumber\\
&&\qquad\quad\  +\frac{i(k^2+1)}{2(k^2-1)}\left( w(\sqrt{mn}-1 ) -\frac{mn_t-m_tn}{2mn}\right)\left(\begin{array}{cc}
0&1\\-1&0\end{array}\right)-\frac{4i k^2}{k^4-1}\left(\begin{array}{cc}0&1\\-1&0\end{array}\right)\nonumber\\
&&\qquad\quad\ +\left(\begin{array}{cc}0&-\frac{ k(k-i)p}{(k^2-1)(k+i)}\sqrt{\frac{n}{m}}+\frac{ k(k+i)q}{(k^2-1)(k-i)}\sqrt{\frac{m}{n}}\\
\frac{ k(k-i)q}{(k^2-1)(k+i)}\sqrt{\frac{m}{n}}-\frac{ k(k+i)p}{(k^2-1)(k-i)}\sqrt{\frac{n}{m}}&0\end{array}\right).\nonumber
\end{eqnarray}

By the conservation law
\begin{equation}
-(\sqrt{mn})_t=(w\sqrt{mn})_x,
\end{equation}
we define a new function
\begin{equation}
\tilde{p}(x,t,k)=-\frac{(k^2-1)}{4 k}\int_x^{\infty}(\sqrt{mn}-1)\mathrm{d}s+
\frac{(k^2-1)}{4k}x-\frac{2 k(k^2-1)}{(k^2+1)^2}t, \label{2.4}
\end{equation}
then
\begin{equation}
\tilde{p}_x=\frac{(k^2-1)}{4 k}\sqrt{mn},\quad
\tilde{p}_t=-\frac{(k^2-1)}{4 k}w\sqrt{mn}-\frac{2 k(k^2-1)}{(k^2+1)^2},
\end{equation}
and $\tilde{p}_{xt}=\tilde{p}_{tx}$.

As $x\to \pm \infty$, $\tilde{U}, \tilde{V}\to 0$, thus $\hat{\Phi}\sim e^{-i\tilde{p}(x,t,k)\sigma_3}$.
Take the transformation
\begin{equation}\label{trs1}
\Phi(x,t,k)=\hat{\Phi}(x,t,k)e^{i\tilde{p}(x,t,k)\sigma_3},
\end{equation}
then $\Phi(x,t,k)\to I,\ x\to\infty$
and $\Phi$ satisfies the new Lax pair
\begin{equation}\label{Philax}
\Phi_x+i\tilde{p}_x[\sigma_3,\Phi]=\tilde{U}\Phi,\quad \Phi_t+i\tilde{p}_t[\sigma_3,\Phi]=\tilde{V}\Phi.
\end{equation}

Integrating the first equation of \eqref{Philax} with respect to $x$ in two directions, we get the Jost solutions $\Phi_{\pm} $ admitting the following Volterra integral equations
\begin{equation}\label{int}
\Phi_{\pm}(x,t,k) =I+\int_{\pm\infty}^x e^{-\frac{i(k^2-1)}{4 k}((x-y)+\int_y^x(\sqrt{mn}-1)\mathrm{d}s)\hat{\sigma_3}}\left[\tilde{U}(y,t,k)\Phi_{\pm}(y,t,k)\right]\mathrm{d}y.
\end{equation}

Denote
	$\Phi_\pm  =\left( \Phi_\pm^{(1)},  \Phi_\pm^{(2)}\right), $
	where  $\Phi_\pm^{(1)}$ and $\Phi_\pm^{(2)} $ are
	the first and second columns of $\Phi_\pm  $, respectively. By the equations \eqref{int}, $\Phi_{\pm} $ have the following properties
\begin{Proposition}\label{prop1}
 Suppose the initial data $p_0, q_0  \in  1+H^{3,2}(\mathbb{R})$,  then we have
\begin{itemize}
  \item[1.] $\Phi_+^{(1)}(k),\ \Phi_-^{(2)}(k)$ are analytical in the lower-half complex plane
  $\mathbb{C}^-$;
   $\Phi_+^{(2)}(k),\ \Phi_-^{(1)}(k)$ are
   analytical in the upper-half complex plane $\mathbb{C}^+$.
  \item[2.]The Jost functions $\Phi_{\pm}(k)$ admit the following symmetries
\begin{equation}
\Phi(k)=\overline{\Phi(\bar{k}^{-1})}=\sigma_1\overline{\Phi(\bar{k})}\sigma_1=\sigma_2\Phi(-k)\sigma_2.\\
\end{equation}
  \item[3.]Asymptotic behaviors:
  As $k\to\pm\infty$,
\begin{equation}
\left(\Phi_-^{(1)}(k),  \Phi_+^{(2)}(  k)\right)\to I,\quad \left(\Phi_+^{(1)}(k),  \Phi_-^{(2)}(k)\right)\to I;
\end{equation}
As $k\to\pm 1$,
\begin{equation}\label{asym2}
\Phi(x,t,k)=\pm\frac{i}{k\mp1}\alpha_\pm(x,t)\left(\begin{array}{cc}
\pm1&1\\-1&\mp1\end{array}\right)+\mathcal{O}(1),
\end{equation}
where $\alpha_{\pm}(x,t)$ are real-valued functions.
\end{itemize}
\end{Proposition}

\subsubsection{Scattering data}
\quad As two solutions of the Lax pair \eqref{lax2-1}, $\hat{\Phi}_{\pm}(k)$ satisfy a linear relation
\begin{equation}
\hat{\Phi}_+(k)=\hat{\Phi}_-(k)S(k),
\end{equation}
together with \eqref{trs1}, we get
\begin{equation}\label{liner}
\Phi_+(k)=\Phi_-(k)e^{-\tilde{p}(x,t,k)\hat{\sigma}_3}S(k).
\end{equation}

By the symmetries of $\Phi_\pm(k)$, we find
\begin{equation}
S(k)=\sigma_1\overline{S(\bar{k})}\sigma_1=\sigma_2S(-k)\sigma_2, \nonumber
\end{equation}
which implies   $S(k)$  in  the form
\begin{equation}
S(k)=\left(\begin{array}{cc}
\overline{a(\bar{k})}&b(k)\\ \overline{b(\bar{k})}&a(k)\end{array}\right)=
\left(\begin{array}{cc}
a(-k)&b(k)\\ -b(-k)&a(k)\end{array}\right).  \nonumber
\end{equation}
From \eqref{liner}, $a(k), b(k)$ can be expressed as
\begin{equation}\label{abphi}
a(k)=\det \left(\Phi_-^{(1)}\ \Phi_+^{(2)}\right),\quad b(k)=e^{2i\tilde{p}}\det \left(\Phi_-^{(2)}\ \Phi_+^{(2)}\right),
\end{equation}
so $a(k)$ is analytic in $\mathbb{C}^+$ by  Proposition \ref{prop1}.

Besides,   $\Phi_{\pm}(k)$ and  $a(k)$ admit the asymptotics
\begin{align}\label{asym1}
&\Phi_{\pm}(k)=I+\mathcal{O}(k^{-1}),\ \ a(k)=1+\mathcal{O}(k^{-1}),\quad k\to \infty, \nonumber\\
& a(k)=\mathcal{O}((k\pm1)),\quad k\to\pm 1.\nonumber
\end{align}

Suppose that $a(k)$ has $N_1$ simple zeros $\mu_1,\cdots,\mu_{N_1}$ on $\{k\in\mathbb{C}^+:0<\arg k<\frac{\pi}{2}, |k|>1\}$, and $N_2$ simple zeros $\nu_1,\cdots,\nu_{N_2}$ on the circle $\{k=e^{i\psi}:0<\psi<\frac{\pi}{2}\}$.
By the  symmetries of $a(k)$, we have
\begin{align}
&a(\pm\mu_n )=a(\pm \bar{\mu}_n)= a(\pm \mu_n^{-1})=
a(\pm\bar{\mu}_n^{-1})= 0,\quad n=1,\cdots,N_1,\nonumber\\
&a(\pm {\nu}_m)=a(\pm \bar{\nu}_m)=0,\quad
m=1,\cdots,N_2. \nonumber
\end{align}
 We denote the zeros of $a(k)$   as
\begin{align}
&\zeta_n=\mu_n, \zeta_{n+N_1}=-\mu_n^{-1}, \zeta_{n+2N_1}=-\bar{\mu}_n, \zeta_{n+3N_1}=\bar{\mu}_n^{-1}, \ \ n=1,\cdots, N_1;\nonumber\\
& \zeta_{m+4N_1}=\nu_m, \zeta_{m+4N_1+N_2}=-\bar{\nu}_m, \ \  m=1,\cdots, N_2.\nonumber
\end{align}
Then   discrete spectrum is $\mathcal{Z}:=\{\zeta_n, \bar{\zeta}_n\}_1^{4N_1+2N_2}$, as shown in Figure \ref{zero}.

 Denote $\mathcal{N}\triangleq\left\lbrace 1,...,4N_1+2N_2\right\rbrace $ as the subscript set   of  all  zeros, and  we fix  a small positive constant $\delta_0$ to give the
  partitions $\Delta,\nabla$ and $\Lambda$  of $\mathcal{N}$   as follows
\begin{equation} \label{devide}
\Delta=\left\lbrace n \in  \mathcal{N}: \text{Im}\theta_n> 0\right\rbrace, \nabla=\left\lbrace n \in  \mathcal{N}: \text{Im}\theta_n< 0\right\rbrace,
\end{equation}
where $\theta_n:=\theta(\zeta_n)$.
To distinguish different  type of zeros, we further give
\begin{align}
	&\Delta_1=\left\lbrace j \in \left\lbrace 1,...,N_1\right\rbrace: \text{Im}\theta(\mu_j)> 0\right\rbrace,
     \nabla_1=\left\lbrace j \in \left\lbrace 1,...,N_1\right\rbrace: \text{Im}\theta(\mu_j)< 0\right\rbrace,
	\nonumber\\
	&\Delta_2=\left\lbrace i \in \left\lbrace 1,...,N_2\right\rbrace: \text{Im}\theta(\nu_i)> 0\right\rbrace,
      \nabla_2=\left\lbrace i \in \left\lbrace 1,...,N_2\right\rbrace: \text{Im}\theta(\nu_i)< 0\right\rbrace.
	\nonumber
\end{align}

\begin{figure}
  \centering
\begin{tikzpicture}[scale=1.1]																	\draw[-latex](-3.5,0)--(3.5,0)node[right]{\textcolor{black} {Re$k$}};

\draw[dotted](0,0)circle(1.5);

\coordinate (a) at (0.8,1.27);
\fill[red] (a) circle (1pt);
\node at (0.9,1.5) {$\nu_m$};
\node at (2.1,1.8) {$\mu_n$};
\coordinate (a) at (-0.8,1.27);
\fill[red] (a) circle (1pt);
\node at (-0.9,1.5) {$-\bar{\nu}_m$};
\coordinate (a) at (0.8,-1.27);
\fill[red] (a) circle (1pt);
\node at (0.9,-1.5) {$\bar{\nu}_m$};
\coordinate (a) at (-0.8,-1.27);
\fill[red] (a) circle (1pt);
\node at (-0.9,-1.5) {$-\nu_m$};

\coordinate (a) at (-1.37,-0.6);
\fill[red] (a) circle (1pt);
\coordinate (a) at (1.37,-0.6);
\fill[red] (a) circle (1pt);
\coordinate (a) at (-1.37,0.6);
\fill[red] (a) circle (1pt);
\coordinate (a) at (1.37,0.6);
\fill[red] (a) circle (1pt);

\coordinate (a) at (0.53,0.53);
\fill[blue] (a) circle (1pt);
\node at (0.83,0.55) {$\bar{\mu}^{-1}_n$};
\coordinate (a) at (-0.53,0.53);
\fill[blue] (a) circle (1pt);
\node at (-0.75,0.55) {$-\mu^{-1}_n$};
\coordinate (a) at (0.53,-0.53);
\fill[blue] (a) circle (1pt);
\node at (0.83,-0.55) {$\mu^{-1}_n$};
\coordinate (a) at (-0.53,-0.53);
\fill[blue] (a) circle (1pt);
\node at (-0.75,-0.53) {$-\bar{\mu}^{-1}_n$};

\coordinate (a) at (1.78,1.78);
\fill[blue] (a) circle (1pt);
\coordinate (a) at (-1.78,1.78);
\fill[blue] (a) circle (1pt);
\node at (-2.2,1.8) {$-\bar{\mu}_n$};
\coordinate (a) at (1.78,-1.78);
\fill[blue] (a) circle (1pt);
\node at (2.1,-1.8) {$\bar{\mu}_n$};
\coordinate (a) at (-1.78,-1.78);
\fill[blue] (a) circle (1pt);
\node at (-2.2,-1.8) {$-\mu_n$};
										\draw[-latex](0,-3)--(0,3)node[above]{ \textcolor{black}{Im$k$}};
\coordinate (C) at (-0.2,2.2);
										\coordinate (D) at (2.2,0.2);
																							
\end{tikzpicture}
  \caption{\footnotesize There are two types of discrete spectrum: ($\textcolor{red}{\bullet} $)  are  distributed on unitary circle $|k|=1$  and ($\textcolor{blue}{\bullet} $)
   are  distributed  on complex  plane  $\mathbb{C}   $  off the  unitary circle $|k|=1$. }\label{zero}
\end{figure}
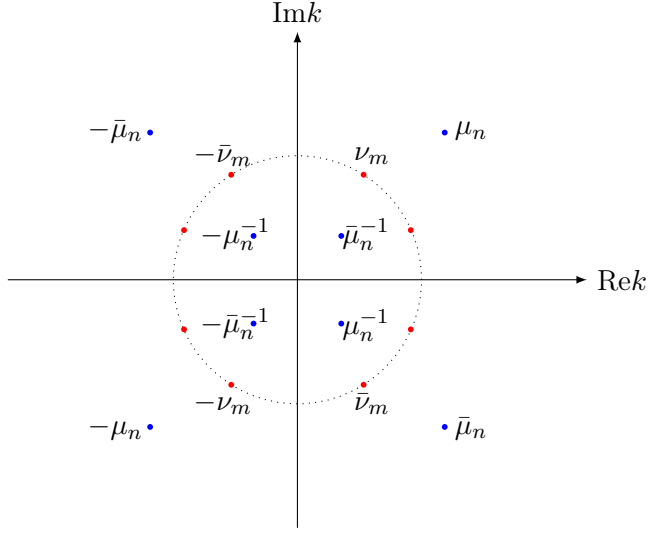

Define the reflection coefficient as
\begin{equation}\label{rk}
r(k)=\frac{\overline{b(k)}}{a(k)},
\end{equation}
which  admits the following symmetries
\begin{equation}
r(k)=-\overline{r(-\bar{k})}=\overline{r(\bar{k}^{-1})}=-r(-k^{-1})
\end{equation}
 with $r(0) = 0$,  since  $k\to 0$, $a(k)\to1$ and $b(k)\to 0$.

\subsubsection{Jost solutions at $\lambda=0$}

\quad We rewrite the Lax pair \eqref{lax2-1} as
\begin{equation}
\left\{\begin{array}{lr}
\hat{\Phi}_x+\left(\frac{i(k^2-1)}{4k}-\frac{ik}{k^2-1}\frac{mn_x-m_xn}{2mn}\right)\sigma_3\hat{\Phi}=\tilde{U}^0\hat{\Phi},\\
\hat{\Phi}_t-\frac{2i k(k^2-1)}{(k^2+1)^2}\sigma_3\hat{\Phi}=\tilde{V}^0\hat{\Phi},
\end{array}\right.
\end{equation}
where
\footnotesize
\begin{eqnarray}
&&\tilde{U}^0=-\frac{(k^2+1)^2}{4 k(k^2-1)}(\sqrt{mn}-1)\sigma_3-\frac{i(k^2+1)}{2(k^2-1)}\left( \sqrt{mn}-1 +\frac{mn_x-m_xn}{2mn}\right)
\left(\begin{array}{cc}0&1\\-1&0\end{array}\right),\nonumber\\
&&\tilde{V}^0=\frac{ik}{k^2-1}\left( w(\sqrt{mn}-1 ) -\frac{mn_t-m_tn}{2mn}\right)\sigma_3+\frac{i k}{k^2-1}\left(
\sqrt{\frac{n}{m}}p+\sqrt{\frac{m}{n}}q-2 \right)\sigma_3
+\frac{i(k^2-1)w}{4k}\sqrt{ mn }\sigma_3\nonumber\\
&&\qquad\quad\  +\frac{i(k^2+1)}{2(k^2-1)}\left( w(\sqrt{mn}-1) -\frac{mn_t-m_tn}{2mn}\right)\left(\begin{array}{cc}
0&1\\-1&0\end{array}\right)-\frac{4i k^2}{k^4-1}\left(\begin{array}{cc}0&1\\-1&0\end{array}\right)\nonumber\\
&&\qquad\quad\ +\left(\begin{array}{cc}0&-\frac{ k(k-i)p}{(k^2-1)(k+i)}\sqrt{\frac{n}{m}}+\frac{ k(k+i)q}{(k^2-1)(k-i)}\sqrt{\frac{m}{n}}\\
\frac{ k(k-i)q}{(k^2-1)(k+i)}\sqrt{\frac{m}{n}}-\frac{ k(k+i)p}{(k^2-1)(k-i)}\sqrt{\frac{n}{m}}&0\end{array}\right).\nonumber
\end{eqnarray}
\small
Take the transformation
\begin{equation}
\Phi^0(k)\equiv\Phi^0(x,t,k)=\hat{\Phi}(x,t,k)e^{\left(\frac{ik}{2(k^2-1)}\log\frac{n}{m}+\frac{i(k^2-1)}{4k}x-\frac{2ik(k^2-1)}{(k^2+1)^2}x\right)\sigma_3},
\end{equation}
then as $x\to\infty$, $\Phi^0(k)\to I$, and $\Phi^0(k)$ admits the following Lax pair
\begin{equation}
\left\{\begin{array}{lr}
\Phi^0_x+\left(\frac{i(k^2-1)}{4k}-\frac{ik}{k^2-1}\frac{mn_x-m_xn}{2mn}\right)[\sigma_3,\Phi^0]=\tilde{U}^0\Phi^0,\\
\Phi^0_t-\frac{2i k(k^2-1)}{(k^2+1)^2}[\sigma_3,\Phi^0]=\tilde{V}^0\Phi^0.
\end{array}\right.
\end{equation}
The matrix function $\Phi^0(k)$ has the following asymptotic behavior as $k\to i$:
\begin{equation}\label{asymi}
\Phi^0(x,t,k)=I+\left(\begin{array}{cc}
f_1(t)&\frac{1}{2 }\sqrt{\frac{m}{n}}q\\
\frac{1}{2 }\sqrt{\frac{n}{m}}p&f_2(t)\end{array}\right)(k-i)
+\mathcal{O}((k-i)^2),
\end{equation}
where $f_1(t), f_2(t)$ are functions only related to $t$.

Since $\Phi(k)e^{-i\tilde{p}(k)\sigma_3}$ and $\Phi^{0}(k)e^{\left(-\frac{ik}{2(k^2-1)}\log\frac{n}{m}-\frac{i(k^2-1)}{4k}x+\frac{2ik(k^2-1)}{(k^2+1)^2}x\right)\sigma_3}$ are two solutions of the same Lax pair \eqref{lax2-1}, they admit the  linear relations as follows:
\begin{equation}\label{line2}
\Phi(k)=\Phi^0(k)e^{\left(\frac{i(k^2-1)}{4 k}\int^x_{\pm\infty}(\sqrt{mn}-1)\mathrm{d}s-
\frac{ik}{2(k^2-1)}\log\frac{n}{m}\right)\sigma_3}.
\end{equation}
Therefore, together \eqref{asymi} with \eqref{line2}, the asymptotic behavior of $\Phi(k)$ at $k=i$ can be deduced.
\subsection{Scattering map }
\quad\ \  In this part, we discuss the relation between the initial values $p_0(x), q_0(x)$ and the reflection coefficient $r(k)$. We first show the main result by the following proposition:

\begin{Proposition}\label{pqr}
For the initial data $p_0, q_0\in 1+H^{3,2}(\mathbb{R})$ satisfying
$$(p_0+p_{0x})(q_0-q_{0x})=m(x,0)n(x,0)>\epsilon_0,$$  
where $\epsilon_0>0$ is a   constant,  the map $\{p_0, q_0\}\to r(k)$ is Lipschitz continuous from $1+H^{3,2}(\mathbb{R})$ to $H^{1,1}(\mathbb{R})$.
\end{Proposition}
We show the primary difficulty in the proof:
Based on the fact that $S(k)$ is independent with $x, t$, by taking $x=t=0$ in equation \eqref{abphi}, we have
\begin{eqnarray}
&&a(k)-1=n_{11}^-(0,k)\overline{n_{11}^+(0,k)}-n_{21}^-(0,k)\overline{n_{21}^+(0,k)}+n_{11}^-(0,k)+\overline{n_{11}^+(0,k)},\label{a-1}\\
&&e^{-2i\tilde{p}(0,0,k)}\overline{b(k)}=n_{11}^-(0,k)n^+_{21}(0,k)-n_{21}^-(0,k)n^+_{11}(0,k)+n^+_{21}(0,k)-n_{21}^-(0,k),\label{eb}
\end{eqnarray}
with
$$\textbf{n}^{\pm}(x,k):=(n^{\pm}_{11}(x,k),n^{\pm}_{21}(x,k))^T=(\phi_{11}^{\pm}(x,k)-1,\phi_{21}^{\pm}(x,k))^T,$$
where $\phi^{\pm}_{jl}(x,k)$ are the components of $\Phi_{\pm}(x,0,k)$. For $k \in \mathbb{R}$, $\tilde{p}(0,0,k)$ is real, thus $\| \overline{b(k)}\|_{L^2(\mathbb{R})}=\| e^{-2i\tilde{p}(0,0,k)}\overline{b(k)}\|_{L^2(\mathbb{R})}$.

Take $\textbf{n}^+(x,k)$ for example and to convenient, we replace it by $\textbf{n}(x,k)$. By \eqref{int}, we have
\begin{equation}\label{n}
\textbf{n}(x,k)=\textbf{n}_0(x,k)+T\textbf{n}(x,k),
\end{equation}
where $\textbf{n}_0(x,k)=T\textbf{e}_1$ and $T$ is an integral operator defined by
\begin{equation}
T\textbf{f}(x,k)=\int_x^{+\infty}K(x,y,k)\textbf{f}(y,k)\mathrm{d}y,
\end{equation}
with the kernel
\begin{equation}\label{K}
K(x,y,k)=\frac{ik\tilde{m}(y)}{k^2-1}\left(\begin{array}{cc}1&0\\e^{-\frac{i}{2}(k-\frac{1}{k})(h(y)-h(x))}&0\end{array}\right)
+\frac{i(k^2+1)\tilde{m}(y)}{2(k^2-1)}\left(\begin{array}{cc}0&1\\-1&e^{-\frac{i}{2}(k-\frac{1}{k})(h(y)-h(x))}\end{array}\right),
\end{equation}
here the new function $\tilde{m}(x)$ is introduced as
\begin{equation}
\tilde{m}(x):=\sqrt{m(x,0)n(x,0)}-1+\frac{m(x,0)n_{x}(x,0)-m_{x}(x,0)n(x,0)}{2m(x,0)n(x,0)},
\end{equation}
and $h(x):=x-\int_x^{\infty}(\sqrt{m(s,0)n(s,0)}-1)\mathrm{d}s$.

To finish the proof, we deduce the following lemma
\begin{lemma}
The map
$
\{p_0(x), q_0(x)\}\to \tilde{m}(x)
$
 is Lipschitz continuous from $1+H^{3,2}(\mathbb{R})$ to $H^{1,2}(\mathbb{R})$.
\end{lemma}

\begin{proof}
For $p_0, q_0\in1+H^{3,2}(\mathbb{R})$, by \eqref{2mch-q}, we have
$m(x,0)-1, n(x,0)-1 \in H^{2,2}(\mathbb{R})$. Therefore
\begin{align}
\int_{\mathbb{R}}\mid x^2\tilde{m}(x)\mid^2\mathrm{d}x=&\int_{\mathbb{R}}\mid\frac{x^2(m(x,0)n(x,0)-1)}{\sqrt{m(x,0)n(x,0)}+1}+\frac{x^2m(x,0)n_{x}(x,0)-x^2m_{x}(x,0)n(x,0)}{2m(x,0)n(x,0)}\mid^2\mathrm{d}x\nonumber\\
\lesssim&\int_{\mathbb{R}}(x^2(m(x,0)-1)n(x,0)+x^2(n(x,0)-1))^2\mathrm{d}x+\int_{\mathbb{R}}(m(x,0)x^2n_{x}(x,0))^2\mathrm{d}x\nonumber\\
&+\int_{\mathbb{R}}(x^2m_{x}(x,0)n(x,0))^2\mathrm{d}x
\lesssim\parallel m-1\parallel_{H^{1,2}(\mathbb{R})}+\parallel n-1\parallel_{H^{1,2}(\mathbb{R})}.\nonumber
\end{align}
By similar calculation, we deduce that
\begin{equation}
\int_{\mathbb{R}}\mid x^2\tilde{m}_x(x)\mid^2\mathrm{d}x\lesssim\parallel m-1\parallel_{H^{2,2}(\mathbb{R})}+\parallel n-1\parallel_{H^{2,2}(\mathbb{R})}.
\end{equation}\end{proof}

Based on above results, $n_{11}(k), n_{22}(k)$ can be proved to belong to $H^1(\mathbb{R})$, and $kb(k)\in L^2(\mathbb{R})$ by the analytical method shown in \cite{mch8,Xu,Ju}, which together with \eqref{rk}, \eqref{a-1} and \eqref{eb} shows that $r(k)\in H^{1,1}(\mathbb{R})$.

\subsection{Set-up of a  RH problem}

\quad  We introduce a new  scale
\begin{equation}
y: =x-\int_x^{+\infty}(\sqrt{mn}-1)\mathrm{d}s, \label{2.39}
\end{equation}
and write   (\ref{2.4}) in the form
\begin{equation}
 \tilde{p}(x,t,k)=t\theta(k,\xi),
\end{equation}
where
\begin{equation}
\theta(k,\xi)=\frac{k^2-1}{4k}\xi-\frac{2k(k^2-1)}{(k^2+1)^2}, \ \ \xi =\frac{y}{t}. \nonumber
\end{equation}

Define a matrix function
\begin{equation}
M(k):=M(y,t,k)=\left\{\begin{array}{lr}
\left(\frac{\Phi_-^{(1)}}{a\left(k\right)}\quad \Phi_+^{(2)}\right),\quad \mathrm{Im}k>0,\\
\left(\Phi_+^{(1)}\quad \frac{\Phi_-^{(2)}}{\overline{a(\bar{k})}}\right),\quad \mathrm{Im}k<0,\end{array}
\right.
\end{equation}
which solves the following RH problem
\begin{RHP}\label{rhp1}
Find a $2\times 2$ matrix-valued function $M(k)$ satisfying \begin{itemize}
             \item Analyticity: $M(k)$ is meromorphic in $\mathbb{C}\setminus\mathbb{R};$
             \item Symmetry: $M(k)=\sigma_1 \overline{ M(\bar{k})}\sigma_1=\sigma_2 M(-k)\sigma_2=\overline{M(\bar{k}^{-1})};$
             \item Jump condition: $M(k)$ has continuous boundary values $M_{\pm}(k)$ on $\mathbb{R}$ and
\begin{equation}
M_+(k)=M_-(k)V(k),\quad k\in \mathbb{R},
\end{equation}
where
\begin{equation}
V(k)=e^{-it\theta(k)\hat{\sigma}_3}V_0(k),
\end{equation}
and
\begin{equation}
V_0(k)=\left(\begin{array}{cc}
1-|r(k)|^2&\bar{r}(k)\\-r(k)&1\end{array}\right);
\end{equation}
             \item Asymptotic behaviors:
\begin{eqnarray}
&&M(k)=I+\mathcal{O}(k^{-1}),\quad k\to \infty;\nonumber\\
&&M(k)=\frac{i}{k-1}\alpha_+\left(\begin{array}{cc}
0&1\\0&-1\end{array}\right)+\mathcal{O}(1),\quad k\to 1,\ \mathrm{Im}k>0;\nonumber\\
&&M(k)=\frac{i}{k-1}\alpha_+\left(\begin{array}{cc}
1&0\\-1&0\end{array}\right)+\mathcal{O}(1),\quad k\to 1,\ \mathrm{Im}k<0;\nonumber\\
&&M(k)=\frac{i}{k+1}\alpha_-\left(\begin{array}{cc}
0&0\\-1&-1\end{array}\right)+\mathcal{O}(1),\quad k\to 1,\ \mathrm{Im}k>0;\nonumber\\
&&M(k)=\frac{i}{k+1}\alpha_-\left(\begin{array}{cc}
1&1\\0&0\end{array}\right)+\mathcal{O}(1),\quad k\to 1,\ \mathrm{Im}k>0;\nonumber\\
&&M(k)=e^{\tau_+\sigma_3}+e^{\tau_+\sigma_3}\left(\begin{array}{cc}
f_1(t)&\frac{1}{2}\sqrt{\frac{m}{n}}q\\
\frac{1}{2}\sqrt{\frac{n}{m}}p&f_2(t)\end{array}\right)(k-i)
+\mathcal{O}((k-i)^2),\quad k\to i, \label{2.45}
\end{eqnarray}
where
\begin{equation}
\tau_+=\frac{1}{2}\int_{-\infty}^{+\infty}(\sqrt{mn}-1 )\mathrm{d}s-\frac{1}{4}\ln\frac{n}{m};
\end{equation}
\item
Residue condition: $M(k)$ has simple poles at each $\zeta_n\in\mathcal{Z}$ with
\begin{eqnarray}
&&\res_{k=\zeta_n}M(k)=\lim_{k\to\zeta_n}M\left(\begin{array}{cc}
0&0\\c_ne^{2it\theta}&0\end{array}\right),\\
&&\res_{k=\bar{\zeta}_n}M(k)=\lim_{k\to\bar{\zeta}_n}M\left(\begin{array}{cc}
0&\bar{c}_ne^{-2it\theta}\\0&0\end{array}\right),
\end{eqnarray}
where $c_n=-\frac{\overline{b(\bar{\zeta_n})}}{a'(\zeta_n)}, n=1,\cdots,4N_1+2N_2.$
           \end{itemize}
\end{RHP}
From the asymptotic expansion  (\ref{2.45}), we deduce the reconstruction formula for the solutions   of the 2-mCH equation \eqref{2mch-p}-\eqref{2mch-q} as follows
\begin{equation}\label{restr}
 \begin{array}{lr}
(\log p)_x=-\partial_x(\log G_3)G_1^{-1}-1,\\[5pt]
(\log q)_x=1-\partial_x(\log G_3)G_2^{-1},\\[5pt]
 \displaystyle{x(y,t)=y +  \lim_{k\to i}\log \frac{M_{11} }{M_{22} } +\log\frac{q-q_x}{p+p_x}},
\end{array}
\end{equation}
where
\begin{eqnarray}
&&G_1=\lim_{k\to i}\frac{M_{21} M_{22} }{M_{11} M_{12} }-1, \ \ \ G_2=\lim_{k\to i}\frac{M_{11} M_{12} }{M_{21} M_{22}  }-1, \ \
 G_3=\lim_{k\to i}\frac{M_{12} M_{21}  }{(k-i)^2}.\nonumber
\end{eqnarray}

\begin{figure}
	\centering \subfigure[No phase point on $\mathbb{R}$\qquad\quad]{\includegraphics[width=0.27\linewidth]{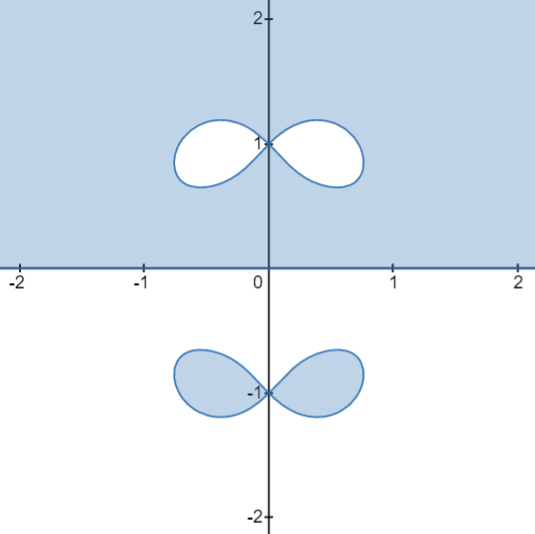}\hspace{1cm}
	\label{fig:desmos-graph}} \subfigure[Four phase points on $\mathbb{R}$\qquad\quad]{\includegraphics[width=0.27\linewidth]{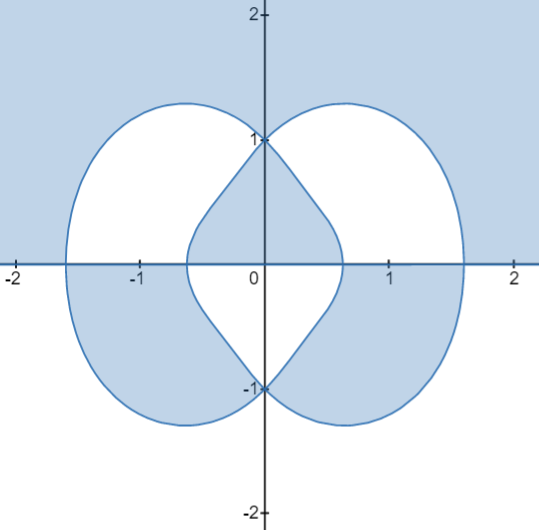}\hspace{1cm}
	\label{fig:desmos-graph-5}}	\subfigure[Eight  phase points on $\mathbb{R}$\qquad\quad]{\includegraphics[width=0.27\linewidth]{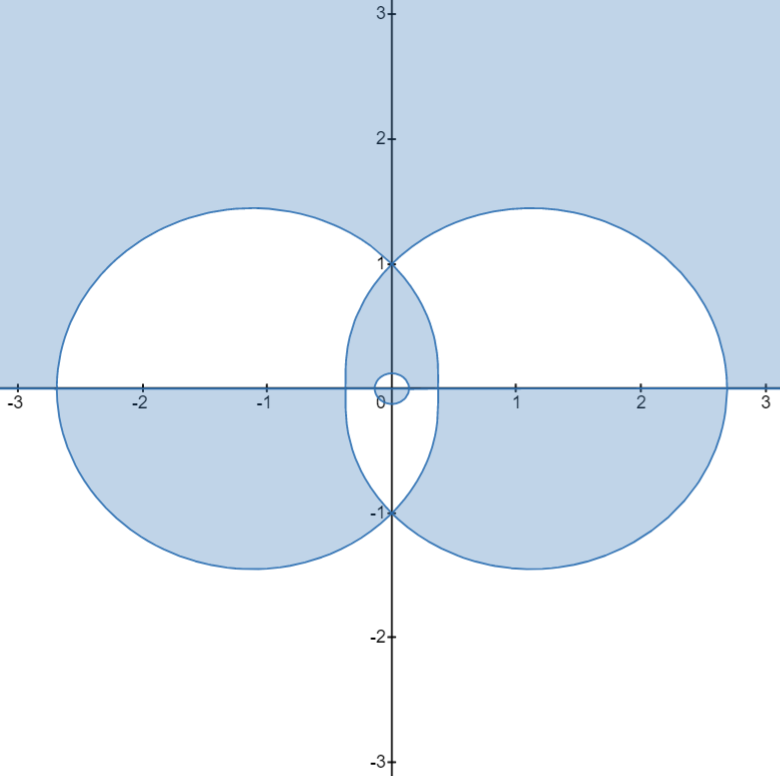}\hspace{1cm}
	\label{fig:desmos-graph-3}} \subfigure[No phase point on $\mathbb{R}$\qquad\quad]{\includegraphics[width=0.27\linewidth]{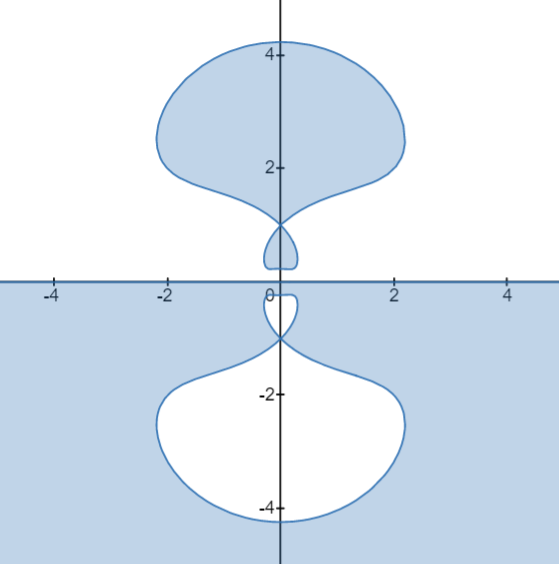}\hspace{1cm}
	\label{fig:desmos-graph-6}}
	\caption{\footnotesize The classification  of $\text{sign Im}\theta$.  In the blue regions, $\text{Im}\theta>0$,  which  implies that $|e^{2it\theta}|\to 0$ as $t\to\infty$.
While  in the white regions,   $\text{Im}\theta<0$, which implies  $|e^{-2it\theta}|\to 0$ as $t\to\infty$.   The blue curves  $\text{Im}\theta=0$ are critical lines between decay and growth regions.  }
	\label{figtheta}
\end{figure}

To remove the singularity of $M(k)$ at $k=\pm 1$, we introduce the following transformation
\begin{equation}\label{trans1}
M(k)=F(k)M^J(k),
\end{equation}
where
$$ F(k):=\left(I+\frac{\sigma_1}{k}\right)^{-1}\left(I+\frac{\sigma_1M^J(0)^{-1}}{k}\right). $$
Then matrix function $M^J(k)$ satisfies a new RH problem:
\begin{RHP}\label{rhp2}
Find a $2\times 2$ matrix-valued function $M^J(k)$ satisfying \begin{itemize}
             \item Analyticity: $M^J(k)$ is meromorphic in $\mathbb{C}\setminus\mathbb{R};$
             \item Symmetry: $M^J(k)=\sigma_1 \overline{ M^J(\bar{k})}\sigma_1=\sigma_2 M^J(-k)\sigma_2=\sigma_1M^J(0)^{-1}M^J(k^{-1})\sigma_1;$
             \item Jump condition: $M(k)$ has continuous boundary values $M^J_{\pm}(k)$ on $\mathbb{R}$ and
\begin{equation}
M^J_+(k)=M^J_-(k)V(k),\quad k\in \mathbb{R};
\end{equation}
\item
Residue condition: $M^J(k)$ has simple poles at each $\zeta_n\in\mathcal{Z}$ with
\begin{eqnarray}
&&\res_{k=\zeta_n}M^J(k)=\lim_{k\to\zeta_n}M^J(k)\left(\begin{array}{cc}
0&0\\c_n e^{2it\theta}&0\end{array}\right),\label{resd-j}\\
&&\res_{k=\overline{\zeta}_n}M^J(k)=\lim_{k\to\overline{\zeta}_n}M^J(k)\left(\begin{array}{cc}
0&\bar{c}_ne^{-2it\theta}\\0&0\end{array}\right),
\end{eqnarray}
where $c_n=-\frac{\overline{b(\bar{\zeta_n})}}{a'(\zeta_n)}, n=1,\cdots,4N_1+2N_2.$
           \end{itemize}
\end{RHP}

The jump matrix  $V(k)$  admits two different  factorizations
	\begin{align}
		V(k)&=\left(\begin{array}{cc}
			1 &  \bar{r} e^{-2it\theta} \\
			0 & 1
		\end{array}\right)\left(\begin{array}{cc}
			1 & 0\\
			-r e^{2it\theta} & 1
		\end{array}\right) \label{V11}\\
		&=\left(\begin{array}{cc}
			1 & 0\\
			-\frac{r}{1-|r|^2}e^{2it\theta}& 1
		\end{array}\right)(1-|r|^2)^{\sigma_3}\left(\begin{array}{cc}
			1 & \frac{\bar{r}}{1-|r|^2}e^{-2it\theta} \\
			 0 & 1
		\end{array}\right).\label{V12}
	\end{align}

To deal with  the  oscillatory terms $e^{\pm2it\theta}$   in the jump matrix and residue conditions of RHP \ref{rhp1},
we consider the real part of $2it\theta$:
\begin{equation}
\mathrm{Re}(2it\theta)=-2t\mathrm{Im} k\left[
\frac{\xi}{4}(1+|k|^{-2})+2\dfrac{|k|^6-2|k|^4-(3\text{Re}^2k-\text{Im}^2k)(1+|k|^2)-2|k|^2
+1}{\left((\text{Re}^2k-\text{Im}^2k+1)^2+4\text{Re}^2k\text{Im}^2k \right)^2 }\right].\nonumber
\end{equation}
The signature of $\mathrm{Im}k$ are shown in Figure \ref{figtheta}.
We divide $(y,t)-$half plane into four asymptotic regions as follows:

\begin{itemize}

\item   I.  $ { {\xi}}>2$: No   phase point on the real axis, which corresponds to  Figure \ref{figtheta} (a);

\item   II. $0< { {\xi}}<2$:   Four     phase points  on the real axis, which corresponds to     Figure  \ref{figtheta} (b);

\item   III.  $-1/4< { {\xi}}<0$:   Eight  phase points  on the real axis, which corresponds to   Figure \ref{figtheta} (c);

\item  IV. $  {\xi}<-1/4$:   No   phase point on the real axis, which corresponds to  Figure \ref{figtheta} (d).

\end{itemize}

\section{Long-time asymptotics in regions without phase point} \label{sec3}

\quad In this section, we study  the   asymptotics for the 2-mCH equation in
two regions  I.  $  \xi  >2$ and IV.\ $ \xi <-1/4 $ which contain no   phase point on the real axis.

\subsection{Normalization of the RH problem}


\quad Denote
\begin{equation}
I( {\xi})=\left\{\begin{array}{ll}\emptyset,\ &\text{for\ } \  \xi>2,\\
\mathbb{R},&\text{for\  } \  \xi <-1/4;\end{array}\right.
\end{equation}
and define a new function
	\begin{equation}\label{tk}
	T(k)=\prod_{n\in\nabla}\frac{ k-\zeta_n}{\zeta_n^{-1}k-1}\exp\left\{ i\frac{1}{2\pi}\int_{I(\xi)}\frac{\log(1-|r|^2)}{s-k}ds\right\},
	\end{equation}
then $T(k)$ admits the jump relation on $\mathbb{R}$
$$T_+(k)=T_-(k) (1-|r|^2),\  \ k\in I(\xi); \ \ T_+(k)=T_- (k),  \ \  k\in \mathbb{R}\setminus I(\xi).$$

Further  we introduce a  new  matrix-valued   function
\begin{equation}
M^{(1)}(k):= M^{(1)}(y,t,k)=M^J(k)T(k)^{\sigma_3},\label{trans2}
\end{equation}
which then satisfies jump  relation
\begin{equation*}
M^{(1)}_+(k)=M^{(1)}_-(k)V^{(1)}(k),\quad k\in\mathbb{R},
\end{equation*}
where
\begin{equation}\label{v1}	
V^{(1)}(k)=\left(\begin{array}{cc}
1 & \bar\rho(k)T_-(k)^{-2}e^{-2it\theta} \\0&1
\end{array}\right)\left(\begin{array}{cc}
1 & 0\\ -{\rho}(k) T_+(k)^2e^{2it\theta} & 1
\end{array}\right),\quad   k\in 	\mathbb{R}.
\end{equation}
with
\begin{equation}
\rho(k) =\left\{\begin{array}{lll}
r(k),\ &  \text{for\  } \  \xi>2 , \\
-\dfrac{r(k) }{1-|r(k)|^2},  & \ \text{for\ } \  \xi <-1/4.
\end{array}\right.
\end{equation}

And $M^{(1)}(k)$ has simple poles at $\zeta_n$ and $\bar{\zeta}_n,\ n=1,\cdots,4N_1+2N_2$ with the following residue conditions
\begin{align}
 &\res_{k=\zeta_n} M^{(1)}(k)=\lim _{k \rightarrow \zeta_n} M^{(1)}(k)\left(\begin{array}{cc}
0 & c_n^{-1}T^{\prime}\left(\zeta_n\right)^{-2} e^{-2 i t \theta_n} \\
0 & 0
\end{array}\right), &n \in \nabla, \label{res1}\\
 &\res_{k=\zeta_n} M^{(1)}(k)=\lim _{k \rightarrow \zeta_n} M^{(1)}(k)\left(\begin{array}{cc}
0 & 0 \\
c_n T^2\left(\zeta_n\right) e^{2 it\theta_n} & 0
\end{array}\right), &n \in \Delta, \\
&\res_{k=\bar{\zeta}_n} M^{(1)}(k)=\lim _{k \rightarrow \bar{\zeta}_n} M^{(1)}(k)\left(\begin{array}{cc}
0 & 0 \\
\bar{c}_n^{-1}\left(\frac{1}{T}\right)^{\prime}\left(\bar{\zeta}_n\right)^{-2} e^{2 i t \bar{\theta}_n} & 0
\end{array}\right), &n \in \nabla \text {, } \\
&\res_{k=\overline{\zeta}_n} M^{(1)}(k)=\lim_{k \rightarrow \bar{\zeta}_n} M^{(1)}(k)\left(\begin{array}{cc}
0 & \bar{c}_n T^{-2}\left(\bar{\zeta}_n\right) e^{-2 i t \bar{\theta}_n} \\
0 & 0
\end{array}\right), &n \in\Delta \text {. } \label{res4}
\end{align}

\subsection{A hybrid $\bar{\partial}$-RH problem}
\qquad In this part,  we make a continuous extension of the jump matrix  $V^{(1)}(k)$
to remove the jump  from $\mathbb{R}$ by using two factorization (\ref{V11})-(\ref{V12}).

Define  two   rays
\begin{align}
&\Sigma_1 =\left\lbrace k\in\mathbb{C}:\ k=\mathbb{R}_+ e^{ i\varphi}\right\rbrace  \cup \left\lbrace k\in\mathbb{C}:\ k=\mathbb{R}_+ e^{ i(\pi-\varphi)}\right\rbrace, \nonumber\\
&\Sigma_2 =\left\lbrace k\in\mathbb{C}:\ k=\mathbb{R}_+ e^{-i\varphi}\right\rbrace  \cup \left\lbrace k\in\mathbb{C}:\ k=\mathbb{R}_+ e^{ i(\pi+\varphi)}\right\rbrace,\nonumber
\end{align}
which together with $\mathbb{R}$ divide the complex plane  $\mathbb{C}$ into  three  regions $\Omega_1$, $\Omega_2$   and  $\Omega_3$, see Figure \ref{figR2}.
Moreover, we open the contour $\mathbb{R}$  with a sufficiently small angle $\varphi$   such that there is no any pole in $\Omega_1$  and  $\Omega_2$.
Denote $\Sigma=\Sigma_1\cup\Sigma_2 $.



\begin{figure}[h]
\begin{center}
\subfigure[\footnotesize   The opened contour $\Sigma$ for the asymptotic region   ${\rm I.}\ \xi>2$, which corresponds to    the   Figure  \ref{figtheta}-(a).       ]{
	\begin{tikzpicture}[scale=1.2]								
										\draw[teal,thick](0,0)--(3,0.5);
										\draw[teal,thick](0,0)--(-3,0.5);
										\draw[brown,thick](0,0)--(-3,-0.5);
										\draw[brown,thick](0,0)--(3,-0.5);
										\draw[-latex,dotted ](-4,0)--(4,0)node[right]{\textcolor{black} {Re$k$}};

\draw[dotted](0,0)circle(1.5);

\coordinate (a) at (0.8,1.27);
\fill[red] (a) circle (1pt);
\node at (1.1,1.25) {$\nu_m$};
\node at (2,1.8) {$\mu_n$};
\coordinate (a) at (-0.8,1.27);
\fill[red] (a) circle (1pt);
\coordinate (a) at (0.8,-1.27);
\fill[red] (a) circle (1pt);
\coordinate (a) at (-0.8,-1.27);
\fill[red] (a) circle (1pt);

\coordinate (a) at (-1.37,-0.6);
\fill[red] (a) circle (1pt);
\coordinate (a) at (1.37,-0.6);
\fill[red] (a) circle (1pt);
\coordinate (a) at (-1.37,0.6);
\fill[red] (a) circle (1pt);
\coordinate (a) at (1.37,0.6);
\fill[red] (a) circle (1pt);

\coordinate (a) at (0.71,0.71);
\fill[blue] (a) circle (1pt);
\coordinate (a) at (-0.71,0.71);
\fill[blue] (a) circle (1pt);
\coordinate (a) at (0.71,-0.71);
\fill[blue] (a) circle (1pt);
\coordinate (a) at (-0.71,-0.71);
\fill[blue] (a) circle (1pt);

\coordinate (a) at (1.78,1.78);
\fill[blue] (a) circle (1pt);
\coordinate (a) at (-1.78,1.78);
\fill[blue] (a) circle (1pt);
\coordinate (a) at (1.78,-1.78);
\fill[blue] (a) circle (1pt);
\coordinate (a) at (-1.78,-1.78);
\fill[blue] (a) circle (1pt);
										\draw[-latex,dotted](0,-2)--(0,2)node[above]{ \textcolor{black}{Im$k$}};
										\draw[brown][-latex,thick](-3,-0.5)--(-1.5,-0.25);
										\draw[teal][-latex,thick](-3,0.5)--(-1.5,0.25);
										\draw[teal][-latex,thick](0,0)--(1.7,0.285)node[above]{\footnotesize \textcolor[rgb]{0.00,0.00,0.00}{$\Sigma_1$}};
										\draw[brown][-latex,thick](0,0)--(1.7,-0.285)node[below]{\footnotesize \textcolor[rgb]{0.00,0.00,0.00}{$\Sigma_2$}};
\node  at (-1.7,0.48) {\footnotesize $\Sigma_1$};
\node  at (-1.7,-0.48) {\footnotesize $\Sigma_2$};
										\coordinate (C) at (-0.2,2.2);
										\coordinate (D) at (2.2,0.2);
		\coordinate (H) at (0.2,1);
		\fill (H) circle (0pt) node[right] {\footnotesize $\Omega_3$};
		\coordinate (M) at (-0.7,-1);
		\fill (M) circle (0pt) node[right] {\footnotesize $\Omega_3$};
										\fill (D) circle (0pt) node[right] {\footnotesize $\Omega_1$};
										\coordinate (J) at (-2.2,-0.2);
										\fill (J) circle (0pt) node[left] {\footnotesize $\Omega_2$};
										\coordinate (k) at (-2.2,0.2);
										\fill (k) circle (0pt) node[left] {\footnotesize $\Omega_1$};
										\coordinate (k) at (2.2,-0.2);
										\fill (k) circle (0pt) node[right] {\footnotesize $\Omega_2$};
										\coordinate (I) at (0.2,0); \coordinate (I) at (0,0);

										\fill (I) 	circle (1pt)node[below] {\footnotesize $O$} ;								
\end{tikzpicture}}
\subfigure[\footnotesize The opened contour $\Sigma$ for the  asymptotic region   ${\rm IV.}\ \xi<-1/4$, which corresponds to  in  the   Figure  \ref{figtheta}-(d).   ]{
	\begin{tikzpicture}[scale=1.2]								
										\draw[brown,thick](0,0)--(3,0.5);
										\draw[brown,thick](0,0)--(-3,0.5);
										\draw[teal,thick](0,0)--(-3,-0.5);
										\draw[teal,thick](0,0)--(3,-0.5);

\draw[dotted](0,0)circle(1.5);

\coordinate (a) at (0.8,1.27);
\fill[red] (a) circle (1pt);
\coordinate (a) at (-0.8,1.27);
\fill[red] (a) circle (1pt);
\coordinate (a) at (0.8,-1.27);
\fill[red] (a) circle (1pt);
\coordinate (a) at (-0.8,-1.27);
\fill[red] (a) circle (1pt);

\coordinate (a) at (-1.37,-0.6);
\fill[red] (a) circle (1pt);
\coordinate (a) at (1.37,-0.6);
\fill[red] (a) circle (1pt);
\coordinate (a) at (-1.37,0.6);
\fill[red] (a) circle (1pt);
\coordinate (a) at (1.37,0.6);
\fill[red] (a) circle (1pt);

\coordinate (a) at (0.71,0.71);
\fill[blue] (a) circle (1pt);
\coordinate (a) at (-0.71,0.71);
\fill[blue] (a) circle (1pt);
\coordinate (a) at (0.71,-0.71);
\fill[blue] (a) circle (1pt);
\coordinate (a) at (-0.71,-0.71);
\fill[blue] (a) circle (1pt);
\node at (1.1,1.25) {$\nu_m$};
\node at (2,1.8) {$\mu_n$};
\coordinate (a) at (1.78,1.78);
\fill[blue] (a) circle (1pt);
\coordinate (a) at (-1.78,1.78);
\fill[blue] (a) circle (1pt);
\coordinate (a) at (1.78,-1.78);
\fill[blue] (a) circle (1pt);
\coordinate (a) at (-1.78,-1.78);
\fill[blue] (a) circle (1pt);

\draw[-latex,dotted ](-4,0)--(4,0)node[right]{\textcolor{black} {Re$k$}};
										\draw[-latex,dotted](0,-2)--(0,2)node[above]{ \textcolor{black}{Im$k$}};
										\draw[teal,thick][-latex](0,0)--(-1.7,-0.285)node[below]{\footnotesize \textcolor[rgb]{0.00,0.00,0.00}{$\Sigma_2$}};
										\draw[brown,thick][-latex](0,0)--(-1.7,0.285)node[above]{\footnotesize \textcolor[rgb]{0.00,0.00,0.00}{$\Sigma_1$}};
										\draw[brown,thick][-latex](3,0.5)--(1.5,0.25);
\node  at (1.7,0.48) {\footnotesize $\Sigma_1$};
\node  at (1.7,-0.48) {\footnotesize $\Sigma_2$};										\draw[teal][-latex](3,-0.5)--(1.5,-0.25);
										\coordinate (C) at (-0.2,2.2);
										\coordinate (D) at (2.2,0.2);
		\coordinate (H) at (0.2,1);
		\fill (H) circle (0pt) node[right] {\footnotesize $\Omega_3$};
		\coordinate (M) at (-0.7,-1);
		\fill (M) circle (0pt) node[right] {\footnotesize $\Omega_3$};
										\fill (D) circle (0pt) node[right] {\footnotesize $\Omega_1$};
										\coordinate (J) at (-2.2,-0.2);
										\fill (J) circle (0pt) node[left] {\footnotesize $\Omega_2$};
										\coordinate (k) at (-2.2,0.2);
										\fill (k) circle (0pt) node[left] {\footnotesize $\Omega_1$};
										\coordinate (k) at (2.2,-0.2);
										\fill (k) circle (0pt) node[right] {\footnotesize $\Omega_2$};
										\coordinate (I) at (0.2,0); \coordinate (I) at (0,0);

										\fill (I) 	circle (1pt)node[below] {\footnotesize $O$} ;								
\end{tikzpicture}
}
\caption{\footnotesize Opening the real axis $\mathbb{R}$ at $k=0$ with sufficient small angle.
 The opened contours   $\Sigma_1$ and   $\Sigma_2$  decay  along  blue region  and   white  region
    in Figure  \ref{figtheta}, respectively.  The discrete spectrum on unitary circle $|k|=1$ denoted by  ($\textcolor{red}{\bullet} $)   and that
    on $\mathbb{C}\setminus \{ k:\ |k|=1\}$  denoted by ($\textcolor{blue}{\bullet} $).    }
								\label{figR2}
\end{center}
\end{figure}
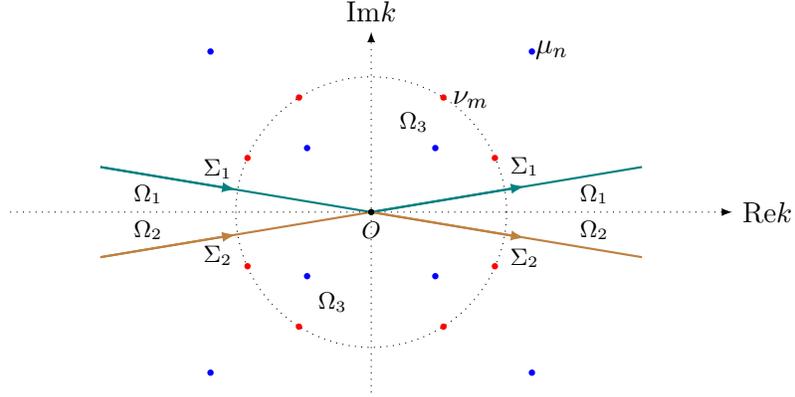
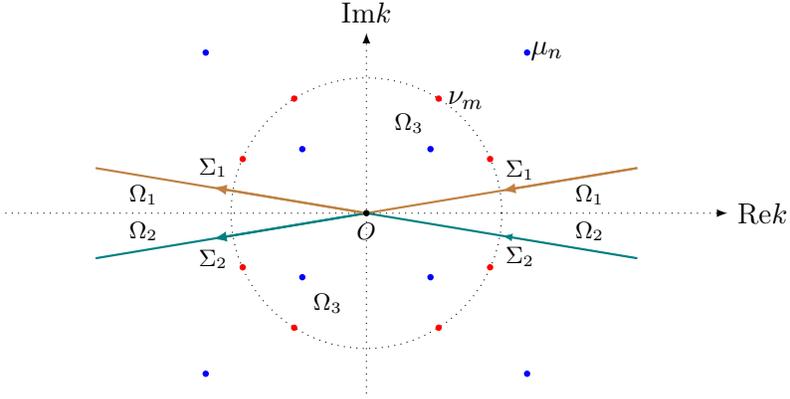

We now introduce the extension functions

\begin{Proposition}\label{proR}
	There are  $R_\ell(k)$: $\bar{\Omega}_\ell\to \mathbb{C}$, $\ell=1,2$  with the following   boundary values
	\begin{align}
&R_1(k)=\Bigg\{\begin{array}{ll}
	 \rho(k)T_+(k)^2 & k\in \mathbb{R},\\
	0  &k\in \Sigma_1,\\
\end{array} \hspace{0.6cm}
R_2(k)=\Bigg\{\begin{array}{ll}
	\bar{\rho}(k)T_-(k)^{-2} &k\in \mathbb{R}, \\
	0 &k\in \Sigma_2,\\
\end{array} \nonumber
\end{align}	
$R_1(k)$ and $R_2(k)$  admit the following estimates, respectively
	\begin{align}
 &|\bar{\partial}R_\ell(k)|\lesssim|r'(|k|)|+|k|^{-1/2}, \  \text{ $ k\in \Omega_\ell$,} \  \ell=1,2,\label{dbarRj}\\
&	\bar{\partial}R_\ell(k)=0,\hspace{0.5cm}  k\in   \Omega_3. \nonumber
	\end{align}

\end{Proposition}

Now we define a new matrix function
\begin{equation}
R^{(2)}(k )=\left\{\begin{array}{lll}
\left(\begin{array}{cc}
1 & 0\\
R_1(k )e^{2it\theta} & 1
\end{array}\right),  &k\in \Omega_1;\\
\\
\left(\begin{array}{cc}
1 & R_2(k )e^{-2it\theta}\\
0 & 1
\end{array}\right), & k\in \Omega_2;\\
\\
I,  &k\in   \Omega_3,\\
\end{array}\right.\label{R(2)-}
\end{equation}
then  new function
\begin{equation}\label{trans3}
M^{(2)}(k):=M^{(1)}(k)R^{(2)}(k)
\end{equation}
  is continuous in $\mathbb{C}\setminus\mathcal{Z}$ satisfying a hybrid $\bar{\partial}$-RH problem
\begin{equation*}
M^{(2)}_+(k)=M^{(2)}_-(k)V^{(2)}(k),\quad k\in \Sigma
\end{equation*}
 with the same residue conditions (\ref{res1})-(\ref{res4}) and  the following $\bar{\partial}$-equation
\begin{equation}
\bar{\partial}M^{(2)}(k)=M^{(2)}(k)\bar{\partial}R^{(2)}(k),
\end{equation}
where
\begin{equation}
	\bar{\partial}R^{(2)}(k )=\left\{\begin{array}{lll}
		\left(\begin{array}{cc}
			0 & 0\\
			\bar{\partial}R_1(k )e^{2it\theta} & 0
		\end{array}\right), & k\in \Omega_1;\\
		\\
		\left(\begin{array}{cc}
			0 & \bar{\partial}R_2(k )e^{-2it\theta}\\
			0 & 0
		\end{array}\right),  &k\in \Omega_2;\\
		\\
		0,  &k\in \Omega_3.\\
	\end{array}\right.
\end{equation}

To solve this hybrid $\bar{\partial}-$RH problem, we decompose it into the form
\begin{equation}\label{m3}
M^{(2)}(k)  =M^{(3)}(k)M^r(k),
\end{equation}
where $M^r(k)$ and $M^{(3)}(k)$  are  the  solutions  of a pure RH problem  and  a pure $\bar{\partial}-$Problem, respectively.
The  pure RH problem  is given as follows.

\begin{RHP}\label{RHP15}
Find a function  $  M^{r}(k)$ with properties
\begin{itemize}
  \item
$M^{r}(k)$ is  meromorphic  in $\mathbb{C};$
\item $M^{r}(k) = I+\mathcal{O}(k^{-1}),\hspace{0.5cm}k \rightarrow \infty;$
\item For $k\in\mathbb{C}$, $\bar{\partial}M^r(k)=0;$
\item $M^{r}(k)$ has the same residue conditions as  (\ref{res1})-(\ref{res4}).

\end{itemize}
\end{RHP}
It has been proved that the RH problem \ref{RHP15} is solvable and there exists an unique solution $M^{r}(k)$
   as shown in Section 7 of \cite{mch8}.

The  $M^{(3)}(k)$  satisfies  a pure $\bar{\partial}$-equation
\begin{equation}
\bar{\partial}M^{(3)}(k)=M^{(3)}(k)W^{(3)}(k),\ \ k\in \mathbb{C},\label{m3}
\end{equation}
where
\begin{equation}
W^{(3)}(k)=M^{r}(k)\bar{\partial}R^{(2)}(k)M^{r}(k)^{-1}.\nonumber
\end{equation}

\subsection{Contribution from a  pure $\bar{\partial}$-problem}
\qquad Now we consider the contribution from  $M^{(3)}(k)$ defined by the $\bar{\partial}$-equation  \eqref{m3}.
The  solution of the $\bar{\partial}$-problem for  $M^{(3)}(k)$ is equivalent to the integral equation
\begin{equation}\label{m3-1}
M^{(3)}(k)=I+\frac{1}{\pi}\iint_\mathbb{C}\dfrac{M^{(3)}(s)W^{(3)} (s)}{s-k}dm(s),
\end{equation}
where $m(s)$ is the Lebesgue measure on $\mathbb{C}$.
Define  the    Cauchy integral  operator
\begin{equation*}\label{ckdef}
S (f)=\frac{1}{\pi}\iint_C\dfrac{f(s)W^{(3)} (s)}{s-k}dm(s).
\end{equation*}
then   (\ref{m3-1})  can be rewritten as
\begin{equation}
 \left(I-S\right)M^{(3)}(k)=I.\label{ped}
\end{equation}

Further it can be shown that
\begin{equation*}
\parallel S\parallel_{L^{\infty}\to L^{\infty}}\lesssim t^{-1/2},
\end{equation*}
which implies that (\ref{ped}) is solvable. As $k\to i$, $M^{(3)}(k)$ has asymptotic expansion
	\begin{equation}
		M^{(3)}(k)=M^{(3)}_0 +M^{(3)}_1 (k-i)+\mathcal{O}((k-i)^{2}), \label{erw}
	\end{equation}
where
\begin{align}
&M^{(3)}_0=I+\frac{1}{\pi}\iint_\mathbb{C}\dfrac{M^{(3)}(s)W^{(3)} (s)}{s-i}dm(s), \\
&M^{(3)}_1 =\frac{1}{\pi}\iint_C\dfrac{M^{(3)}(s)W^{(3)} (s)}{(s-i)^2}dm(s),
\end{align}
whose estimates  are  derived as  follows

\begin{Proposition}\label{estm2}
	There exists a  constant $\kappa <1/4 $, such that
	\begin{align}
		\parallel M^{(3)}(i)-I\parallel \lesssim t^{-1+\kappa}, \  \ \ |M^{(3)}_1(y,t)|\lesssim t^{-1+\kappa}.  \label{m3i}
	\end{align}

\end{Proposition}

\subsection{Proof of Theorem 1--Case I. }
\qquad

We  now  construct the long-time asymptotics for   the 2-mCH equation \eqref{2mch-p}-\eqref{2mch-q}
to complete the proof of Theorem 1--Case I.

Inverting all above transformations (\ref{trans1}), (\ref{trans2}), (\ref{trans3}) and (\ref{m3}), we have
\begin{align}
M(k)=&F(k)M^{(3)}(k) M^r(k)R^{(2)}(k)^{-1}T(k)^{-\sigma_3}. \label{ope}
\end{align}
To  reconstruct   $p(x,t), q(x,t)$ by using \eqref{restr},  we  take    $k\to i$  along the path in $ \Omega_3$.
In this case,  $ R^{(2)}(k)=I$.

Consider the  asymptotic expansions  of $F(k)$  and $T(k)$ at $k=i$
\begin{align}
&F(k)=F_1+F_2(k-i)+O((k-i)^2),\label{fasym}\\
&T(k)=T_0+T_1(k-i)+\mathcal{O}((k-i)^2).\label{tasym}
\end{align}
 Further substituting  expansions (\ref{erw}), (\ref{fasym}) and (\ref{tasym}) into (\ref{ope}),
   we obtain that
\begin{align}			
M(k)=M_0 +M_1 (k-i)+\mathcal{O}((k-i)^2),			
\end{align}				
where
$$
M_0(y,t)=F_1M^{r}(i)T_0^{-\sigma_3}+\mathcal{O}(t^{-1+\kappa}),\quad M_1(y,t)=F_2M^{r}(i)T_0^{-\sigma_3}+F_1M^{r}(i)T_1^{-\sigma_3}+\mathcal{O}(t^{-1+\kappa}).
$$

By using the reconstruction formula \eqref{restr}, we further get the long-time asymptotic behavior for   the 2-mCH equation \eqref{2mch-p}-\eqref{2mch-q} as follows
\begin{align}
&(\log p)_x=p^{sol}(x,t)+O(t^{-1+\kappa}),\nonumber\\
&(\log q)_x=q^{sol}(x,t)+O(t^{-1+\kappa}),\nonumber
\end{align}
with
\begin{equation}
x(y,t)=y +S_1(y,t)+\mathcal{O}(t^{-1+\kappa}),
\end{equation}
where
\begin{align}
&p^{sol}(x,t)=- f_0(x,t)\frac{M_{0,11}M_{0,12}}
{M_{0,21}M_{0,22}-M_{0,11}M_{0,12}}-1,\label{pr1}\\
&q^{sol}(x,t)=1-f_0(x,t)\frac{M_{0,21}M_{0,22}} {M_{0,11}M_{0,12}-M_{0,21}M_{0,22}},\label{qr1}\\
&S_1(y,t)=\log M_{0,11}-\log M_{0,22}+\log(q^{sol}-q^{sol}_x)-\log(p^{sol}+p^{sol}_x),\label{s1}
\end{align}
and
\begin{equation*}
f_0(x,t)=\partial_x\log \left(M_{1,12}M_{1,21}\right).
\end{equation*}
Therefore, we prove the first item of Theorem \ref{last}.

\section{Long-time asymptotics in regions with  phase points}\label{sec4}

\qquad In this section, we study the long-time asymptotics for the 2-mCH equation \eqref{2mch-p}-\eqref{2mch-q} in the space-time  regions  II. $0< {\xi}<2$  and   III. $-1/4< {\xi}<0$.

\subsection{Normalization  of the RH problem}

\qquad
In  the regions  II and   III,  there are 4 and 8  stationary phase points on $\mathbb{R}$, respectively.
For convenience, we introduce a function to denote  the number of stationary phase points
\begin{align}
& n( {\xi})=\left\{\begin{array}{l}
4,\quad \text{for\   } 0< {\xi}<2, \\[4pt]
8,\quad \text{for\  } -1/4< {\xi}<0.
\end{array}\right.
\end{align}
In  the function $T(k)$   defined by   \eqref{tk}, take
\begin{align}
& I( {\xi})=\left\{\begin{array}{ll}
\left(k_1, k_2\right) \cup\left(k_3, k_4\right),& \text{for\ } 0< {\xi}<2, \\[4pt]
\left(k_2, k_3\right) \cup\left(k_4, k_5\right) \cup\left(k_6, k_7\right), &\text{for\ } -1/4< {\xi}<0.
\end{array}\right.
\end{align}

Then $T(k)$ admits the following property:\\
As $k\to k_j$ along the ray $k_j+e^{i\phi}\mathbb{R}^+$ with $|\phi|<\pi$, we have  the estimate
	\begin{align}
	\big|T(k)-T_0(k_j)\left[ \eta( k_j)(k-k_j)\right]^{\eta( k_j) i\nu(k_j)}\big|\lesssim \parallel r\parallel_{H^{1}(\mathbb{R})}|k-k_j|^{1/2},
	\end{align}
	where
	\begin{align}
&	T_0(k)=\prod_{n\in \Delta}\dfrac{k-\zeta_n}{\bar{\zeta}_n^{-1}k-1}e^{i\beta(k,k_j)},\nonumber \\
&\beta(k,k_j)=-\eta(k_j)\log\left( \eta(k_j)(k-k_j)+1\right) \nu(k_j)+\int_{I(\xi)}\frac{\nu(s)}{s-k}ds,\nonumber
	\end{align}
 and for  $ j=1,\cdots, n(\xi)$, the sign function  $\eta(  k_j)$ is defined as
\begin{align}
	\eta(  k_j)=\left\{ \begin{array}{ll}
		(-1)^{j+1},   &\text{for\ \ } 0< {\xi}<2,\\[4pt]
		(-1)^{j},   &  \text{for\ \  } -1/4< {\xi}<0.
	\end{array}\right.\label{eta}
\end{align}

We make  a transformation
\begin{equation}
M^{(1)}(k): =M^J(k)T(k)^{\sigma_3},\label{trans2-1}
\end{equation}
then  $M^{(1)}(k)$  satisfies the residue conditions \eqref{res1}-\eqref{res4}
 and   the jump condition
\begin{equation*}
M^{(1)}_+(k)=M^{(1)}_-(k)V^{(1)}(k),\quad k\in\mathbb{R},
\end{equation*}
where
\begin{equation}\label{v2}	
V^{(1)}(k)=\left(\begin{array}{cc}
1 & \bar\rho(k)T_-(k)^{-2}e^{-2it\theta}\\0&1
\end{array}\right)\left(\begin{array}{cc}
1 & 0\\ -{\rho}(k)T_+(k)^2e^{2it\theta}  & 1
\end{array}\right),
\end{equation}
in which the reflection coefficient is defined as
\begin{equation}
\rho(k,\xi) =\left\{\begin{array}{lll}
r(k),& k\in 	\mathbb{R}\setminus I(\xi),\\
-\dfrac{r(k) }{1-|r(k)|^2},&k\in I(\xi).
\end{array}\right.
\end{equation}


\begin{figure}
\begin{center}
	\subfigure[  The opened contour $\Sigma  $ for the asymptotic region   ${\rm II.}\   0<\xi<2$, which corresponds to  in  the   Figure  \ref{figtheta}-(b).  There
are four  phase points on $ \mathbb{R}$.  ]{
\begin{tikzpicture}[scale=0.9]
\draw[-latex,dotted](-5,0)--(5,0)node[right]{ \textcolor{black}{Re$k$}};
\draw[-latex,dotted ](0,-2.5)--(0,2.5)node[right]{\textcolor{black}{Im$k$}};
\draw [brown,thick] (-4,-0.6) to [out=0,in=180] (-2,0.6)
to [out=0,in=180] (0,-0.6) to [out=0,in=180] (2,0.6)  to  [out=0,in=180] (4,-0.6);

\draw [teal,thick](-4,0.6) to [out=0,in=180] (-2,-0.6)
to [out=0,in=180] (0,0.6) to [out=0,in=180] (2,-0.6)  to [out=0,in=180] (4,0.6);
\draw[-latex,teal,thick](-3.2,0.21)--(-3.15,0.16);
\draw[-latex,brown,thick](-3.2,-0.21)--(-3.15,-0.16);
\draw[-latex,teal,thick ](-2.8,-0.21)--(-2.85,-0.15);
\draw[-latex,brown,thick](-2.8,0.21)--(-2.85,0.15);
\draw[dotted](0,0)circle(1.5);
\draw[-latex,brown,thick](-1.2,0.21)--(-1.3,0.3);
\draw[-latex,teal,thick](-1.2,-0.21)--(-1.3,-0.3);
\draw[-latex,brown,thick](-0.8,-0.21)--(-0.7,-0.3);
\draw[-latex,teal,thick](-0.8,0.21)--(-0.7,0.3);

\draw[-latex,teal,thick](0.8,0.21)--(0.85,0.16);
\draw[-latex,brown,thick](0.8,-0.21)--(0.85,-0.16);
\draw[-latex,teal,thick](1.2,-0.21)--(1.15,-0.15);
\draw[-latex,brown,thick](1.2,0.21)--(1.15,0.15);

\draw[-latex,brown,thick](2.8,0.21)--(2.7,0.3);
\draw[-latex,teal,thick](2.8,-0.21)--(2.7,-0.3);
\draw[-latex,brown,thick](3.2,-0.21)--(3.3,-0.3);
\draw[-latex,teal,thick](3.2,0.21)--(3.3,0.3);

\coordinate (a) at (0.71,0.71);
\fill[blue] (a) circle (1.2pt);
\coordinate (a) at (-0.71,0.71);
\fill[blue] (a) circle (1.2pt);
\coordinate (a) at (0.71,-0.71);
\fill[blue] (a) circle (1.2pt);
\coordinate (a) at (-0.71,-0.71);
\fill[blue] (a) circle (1.2pt);


\coordinate (a) at (1.3,0.78);
\fill[red]  (a) circle (1.2pt);
\coordinate (a) at (-1.3,0.78);
\fill[red]  (a) circle (1.2pt);
\coordinate (a) at (0.67,1.35);
\fill[red]  (a) circle (1.2pt);
\coordinate (a) at (-0.67,1.35);
\fill[red]  (a) circle (1.2pt);

\coordinate (a) at (1.3,-0.78);
\fill[red]  (a) circle (1.2pt);
\coordinate (a) at (-1.3,-0.78);
\fill[red]  (a) circle (1.2pt);
\coordinate (a) at (0.67,-1.35);
\fill[red]  (a) circle (1.2pt);
\coordinate (a) at (-0.67,-1.35);
\fill[red]  (a) circle (1.2pt);

\coordinate (a) at (1.77,1.77);
\fill[blue] (a) circle (1.2pt);
\coordinate (a) at (-1.77,1.77);
\fill[blue] (a) circle (1.2pt);
\coordinate (a) at (1.77,-1.77);
\fill[blue] (a) circle (1.2pt);
\coordinate (a) at (-1.77,-1.77);
\fill[blue] (a) circle (1.2pt);
\node at (1,1.5) {$\nu_m$};
\node at (2.1,1.8) {$\mu_n$};

\node  at (4.3,0.6) {$\Sigma_1$};
\node  at (4.3,-0.6) {$\Sigma_2$};
\coordinate (A) at (1,0);
\fill (A) circle (1.2pt) node[right]{$k_2$};
\coordinate (B)  at (3,0);
\fill (B) circle (1.2pt) node[right]{$k_1$};
\coordinate (C)  at (-1,0);
\fill (C) circle (1.2pt) node[right]{$k_3$};
\coordinate (D)  at (-3,0);
\fill (D) circle (1.2pt) node[right]{$k_4$};
\coordinate (I) at (0,0);
		\fill[red] (I) circle (1.2pt) node[below] {$0$};
\end{tikzpicture}
		\label{case1}}
	\subfigure[ The opened contour $\Sigma  $ for the asymptotic region   ${\rm III.}\  -1/4<\xi<0$,  which  corresponds to  in  the   Figure  \ref{figtheta}-(c). There
are eight   phase points on $ \mathbb{R}$. ]{
		\begin{tikzpicture}[scale=0.9]
\draw[-latex,dotted](-6.5,0)--(6.8,0)node[right]{ \textcolor{black}{Re$k$}};
\draw[-latex,dotted  ](0,-3)--(0,3)node[right]{\textcolor{black}{Im$k$}};
\draw [brown,thick] (-6,-0.6)to [out=0,in=180](-4.5,0.6)to [out=0,in=180](-3,-0.6) to [out=0,in=180] (-1.5,0.6)
to [out=0,in=180] (0,-0.6) to [out=0,in=180] (1.5,0.6)  to  [out=0,in=180] (3,-0.6) to [out=0,in=180] (4.5,0.6) to
[out=0,in=180] (6,-0.6);
\draw [teal,thick](-6,0.6)to [out=0,in=180](-4.5,-0.6)to [out=0,in=180](-3,0.6) to [out=0,in=180] (-1.5,-0.6)
to [out=0,in=180] (0,0.6) to [out=0,in=180] (1.5,-0.6)  to [out=0,in=180] (3,0.6) to [out=0,in=180] (4.5,-0.6) to  [out=0,in=180] (6,0.6);
\draw[-latex,teal,thick](-5.5,0.35)--(-5.4,0.23);
\draw[-latex,brown,thick](-5,0.36)--(-5.1,0.23);
\draw[-latex,brown,thick](-3.9,0.23)--(-4,0.35);
\draw[-latex,teal,thick](-3.6,0.23)--(-3.5,0.36);
\draw[-latex,teal,thick](-2.5,0.35)--(-2.4,0.23);
\draw[-latex,brown,thick](-2,0.36)--(-2.1,0.23);
\draw[-latex,brown,thick](-0.9,0.23)--(-1,0.35);
\draw[-latex,teal,thick](-0.6,0.23)--(-0.5,0.36);
\node at (1.5,2.3) {$\nu_m$};
\node at (3,2.7) {$\mu_n$};
\draw[-latex,teal,thick](5.4,0.23)--(5.5,0.35);
\draw[-latex,brown,thick](5.1,0.23)--(5,0.36);
\draw[-latex,brown,thick](4,0.35)--(3.9,0.23);
\draw[-latex,teal,thick](3.5,0.36)--(3.6,0.23);
\draw[-latex,teal,thick](2.4,0.23)--(2.5,0.35);
\draw[-latex,brown,thick](2.1,0.23)--(2,0.36);
\draw[-latex,brown,thick](1,0.35)--(0.9,0.23);
\draw[-latex,teal,thick](0.5,0.36)--(0.6,0.23);

\draw[-latex,brown,thick](-5.5,-0.35)--(-5.4,-0.23);
\draw[-latex,teal,thick](-5,-0.36)--(-5.1,-0.23);
\draw[-latex,teal,thick](-3.9,-0.23)--(-4,-0.35);
\draw[-latex,brown,thick](-3.6,-0.23)--(-3.5,-0.36);
\draw[-latex,brown,thick](-2.5,-0.35)--(-2.4,-0.23);
\draw[-latex,teal,thick](-2,-0.36)--(-2.1,-0.23);
\draw[-latex,teal,thick](-0.9,-0.23)--(-1,-0.35);
\draw[-latex,brown,thick](-0.6,-0.23)--(-0.5,-0.36);
\draw[dotted](0,0)circle(2.5);
\coordinate (a) at (-1.07,-1.07);
\fill[blue] (a) circle (1.2pt);
\coordinate (a) at (1.07,-1.07);
\fill[blue] (a) circle (1.2pt);
\coordinate (a) at (-1.07,1.07);
\fill[blue] (a) circle (1.2pt);
\coordinate (a) at (1.07,1.07);
\fill[blue] (a) circle (1.2pt);

\coordinate (a) at (1.2,2.19);
\fill[red]  (a) circle (1.2pt);
\coordinate (a) at (-1.2,2.19);
\fill[red]  (a) circle (1.2pt);
\coordinate (a) at (-1.2,-2.19);
\fill[red]  (a) circle (1.2pt);
\coordinate (a) at (1.2,-2.19);
\fill[red]  (a) circle (1.2pt);

\coordinate (a) at (2.245,1.1);
\fill[red]  (a) circle (1.2pt);
\coordinate (a) at (-2.245,-1.1);
\fill[red]  (a) circle (1.2pt);
\coordinate (a) at (2.245,-1.1);
\fill[red]  (a) circle (1.2pt);
\coordinate (a) at (-2.245,1.1);
\fill[red]  (a) circle (1.2pt);

\coordinate (a) at (-2.67,-2.67);
\fill[blue]  (a) circle (1.2pt);
\coordinate (a) at (2.67,-2.67);
\fill[blue]  (a) circle (1.2pt);
\coordinate (a) at (-2.67,2.67);
\fill[blue]  (a) circle (1.2pt);
\coordinate (a) at (2.67,2.67);
\fill[blue]  (a) circle (1pt);
\draw[-latex,brown,thick](5.4,-0.23)--(5.5,-0.35);
\draw[-latex,teal,thick](5.1,-0.23)--(5,-0.36);
\draw[-latex,teal,thick](4,-0.35)--(3.9,-0.23);
\draw[-latex,brown,thick](3.5,-0.36)--(3.6,-0.23);
\draw[-latex,brown,thick](2.4,-0.23)--(2.5,-0.35);
\draw[-latex,teal,thick](2.1,-0.23)--(2,-0.36);
\draw[-latex,teal,thick](1,-0.35)--(0.9,-0.23);
\draw[-latex,brown,thick](0.5,-0.36)--(0.6,-0.23);
\node  at (6.3,0.6) {$\Sigma_1$};
\node  at (6.3,-0.6) {$\Sigma_2$};
\coordinate (A) at (3.75,0);
\fill (A) circle (1pt) node[right]{$k_2$};
\coordinate (B)  at (5.25,0);
\fill (B) circle (1pt) node[right]{$k_1$};
\coordinate (C)  at (2.25,0);
\fill (C) circle (1pt) node[right]{$k_3$};
\coordinate (D)  at (0.75,0);
\fill (D) circle (1pt) node[right]{$k_4$};
\coordinate (I) at (0,0);
		\fill[red] (I) circle (1pt) node[below] {$0$};
\coordinate (E) at (-0.75,0);
		\fill (E) circle (1pt) node[right] {$k_5$};
\coordinate (F) at (-3.75,0);
		\fill (F) circle (1pt) node[right] {$k_7$};
\coordinate (G) at (-5.25,0);
		\fill (G) circle (1pt) node[right] {$k_8$};
\coordinate (H) at (-2.25,0);
		\fill (H) circle (1pt) node[right] {$k_6$};
		\end{tikzpicture}
		\label{case2}}
	\caption{\footnotesize
Opening the real axis $\mathbb{R}$ at phase points  $k_j, \ j=1,\cdots, n(\xi)$ with sufficient small angle.
 The opened contours   $\Sigma_1$ and   $\Sigma_2$  decay  in   blue region  and   white  region
    in Figure  \ref{figtheta}, respectively.  The discrete spectrum on unitary circle $|k|=1$ denoted by  ($\textcolor{red}{\bullet} $)   and
    that
    on $\mathbb{C}\setminus \{ k:\ |k|=1\}$  denoted by ($\textcolor{blue}{\bullet} $).    }
	\label{Fig5}
\end{center}
\end{figure}

\subsection{A hybrid $\bar{\partial}$-RH problem}

\quad\ In this part, we make a continuous extension of $V^{(1)}(k)$ to remove the jump from $\mathbb{R}$.
We open the contour $\mathbb{R}$  in the vicinity with  deformation contours  $\Sigma_1 $  and   $\Sigma_2$ as shown in Figure \ref{Fig5},
such that there is  no  any point in the opened regions  $\Omega_1 $ and $\Omega_2$,
which   denote  the regions surrounded  by  the real axis $\mathbb{R}$ with  $\Sigma_1 $  and   $\Sigma_2$,   respectively.

In a similar way to   \cite{mch8},  it can be shown the following proposition.
\begin{Proposition}\label{prod}
There exist  the functions $R_{\ell}(k)$: $\bar{\Omega}_{\ell}\to \mathbb{C}$, $\ell=1,2$  with the boundary values.
	\begin{align} &R_{1}(k,\xi)=\left\{\begin{array}{ll}
	\rho(k)T_+(k)^2, &\hspace{0.4cm}k\in \mathbb{R},\\[4pt]
 \rho(k_j)T_0(k_j)^2 \left[ \eta( k_j)(k-k_j)\right]^{2i\nu(k_j)},  &\hspace{0.4cm}k\in \Sigma_{1},\\
	\end{array}\right. \\[5pt]
	&R_{2}(k,\xi)=\left\{\begin{array}{ll} \bar{\rho}(k)T_-(k)^{-2},
 &k\in  \mathbb{R},\\[4pt] \bar{\rho}(k_j)T_0(k_j)^{-2}\left[ \eta( k_j)(k-k_j)\right]^{-2i\nu(k_j)}, &k\in \Sigma_{2},
	\end{array} \right.
	\end{align}	
where $j=1,\cdots, n(\xi)$.
The functions  $R_{\ell}(k), \ell=1,2$  admit the following estimates:
	\begin{align}
	&|R_{\ell}(k,\xi)|\lesssim \sin^2(\frac{\pi\arg(k-k_j)}{2\varphi})+ \left(1+ \text{Re}(k)^2\right) ^{-1/2}, \quad       k\in \Omega_{\ell},\label{R}\\
	&|\bar{\partial}R_{\ell}(k, \xi)|\lesssim|r'(\text{Re}k)|+|k-k_j|^{-1/2}, \quad     k\in \Omega_{\ell},  \label{dbarRj3}\\
&\bar{\partial}R_{\ell}(k,\xi)=0,\hspace{0.5cm}  elsewhere.\nonumber
	\end{align}
where $\ell=1,2;   \   j=1,\cdots, n(\xi)$.
\end{Proposition}

Define  a new  function
\begin{equation}
R^{(2)}(k,\xi)=\left\{\begin{array}{lll}
\left(\begin{array}{cc}
1 & R_{1}(k,\xi)e^{-2it\theta}\\
0 & 1
\end{array}\right), & k\in \Omega_{1};\\
\\
\left(\begin{array}{cc}
1 & 0\\
R_{2}(k,\xi)e^{2it\theta} & 1
\end{array}\right),  &k\in \Omega_{2};\\
\\
I,  &elsewhere;\\
\end{array}\right.\label{R(2)1}
\end{equation}
where   the functions $R_{j}(k,\xi)$, $j=1,2$ are given by  Proposition \ref{prod}.

Make  a transformation
\begin{equation}\label{trans3-1} M^{(2)}(k):=M^{(2)}(y,t,k)=M^{(1)}(k)R^{(2)}(k),
\end{equation}
then $M^{(2)}(k)$  is a hybrid $\bar{\partial}-$RH problem, which  satisfies the residue conditions \eqref{res1}-\eqref{res4} and  the following   jump condition
\begin{equation}\label{v2}
M^{(2)}_+(k)=M^{(2)}_-(k)V^{(2)}(k),\quad k\in\Sigma^{(2)},
\end{equation}
where
\begin{equation}
V^{(2)}(k)=\left\{\begin{array}{ll}
R^{(2)}(k)|_{\Sigma_{1}},&\text{as } 	k\in\Sigma_{1};\\[12pt]
R^{(2)}(k)^{-1}|_{\Sigma_{2}},&\text{as } 	k\in\Sigma_{2};
\end{array}\right.
\end{equation}
and $M^{(2)}(k)$ also admits  the $\bar{\partial}-$derivative
\begin{equation}
\bar{\partial}M^{(2)}(k)=M^{(2)}(k)\bar{\partial}R^{(2)}(k).
\end{equation}

We decompose  $M^{(2)}(k)$ in the form
\begin{equation}\label{trans4-1}
M^{(2)}(k) =M^{(3)}(k)M^R(k),
\end{equation}
where  $M^R(k)$ is the  following  pure RH problem
\begin{RHP}\label{RHP5}
Find a function  $M^R(k):=M^R(y,t,k)$ with properties:
\begin{itemize}
  \item $M^{R}(k)$ is  meromorphic  in $\mathbb{C}\setminus \Sigma^{(2)}$;
  \item $M^{R}(k)$ satisfies the jump condition:
\begin{equation} M^{R}_+(k)=M^{R}_-(k)V^{(2)}(k),\hspace{0.5cm}k \in \Sigma^{(2)};
\end{equation}
  \item $
	M^{R}(k) = I+\mathcal{O}(k^{-1}),\ k \rightarrow \infty;$
\item For $ k\in \mathbb{C}$,  $\bar{\partial}M^{R}(k)=0$;
\item $M^{R}(k)$ has   the same kind of residue conditions  as \eqref{res1}-\eqref{res4}.
\end{itemize}	
\end{RHP}

 $M^{(3)}(k)$  satisfies  the following   pure $\bar{\partial}$-problem.

\noindent\textbf{ Pure $\bar{\partial}$-problem}. Find a matrix function $  M^{(3)}(k)$ with the following properties:
\begin{itemize}
  \item $M^{(3)}(k)$ is continuous   and has sectionally continuous first partial derivatives in $\mathbb{C}$;
  \item $M^{(3)}(k) \sim I+\mathcal{O}(k^{-1}),\ k \rightarrow \infty;$
  \item For $k\in \mathbb{C},$\
$\bar{\partial}M^{(3)}(k)=M^{(3)}(k)W^{(3)}(k),$
where
\begin{equation}
W^{(3)}(k)=M^{R}(k)\bar{\partial}R^{(2)}(k)M^{R}(k)^{-1}.
\end{equation}
\end{itemize}

\subsection{Contribution from  the pure RH problem}
\qquad

Define
$$ \varrho=\frac{1}{2} \min  \left \lbrace \min_{ j\neq l\in \mathcal{N}  }  |k_l-k_j|, \ \min_{ j \in \mathcal{N}  } |{\rm Im } k_j|,\  \min_{ j \in \mathcal{N}, {\rm Im }\theta(k)=0}  |k_k-k|   \right  \rbrace,$$
and   a  neighborhood of the stationary phase points
\begin{equation}
U(n(\xi))=\underset{j=1,...,n(\xi)}{\bigcup}U_{k_j},\quad U_{k_j}= \left\lbrace k:|k-k_j|< \varrho  \right\rbrace,    \    j=1,...,n(\xi).
\end{equation}
Then outside of $ U(n(\xi))$, the jump matrix $V^{(2)}(k)$  has the following  estimate.
\begin{Proposition}\label{v2pr}
	For $1\leq  p\leq+\infty$, there exists a positive constant $c>0$   such  that the jump matrix    admits the  estimate
	\begin{align}
	\parallel V^{(2)}(k)-I\parallel_{L^{p}(\Sigma_{j}\setminus U(n(\xi)) )}= \mathcal{O}( e^{-c t}), \ \  j=1,2.
	\end{align}

\end{Proposition}

Due to  Proposition \ref{v2pr}, we   construct the solution $M^{R}(k)$ with the following form
\begin{equation}\label{trans5-1}
M^{R}(k)=\left\{\begin{array}{ll}
E(k)M^{r}(k), & k\notin U(n(\xi)),\\[4pt]
E(k)M^{r}(k)M^{lo}(k),  &k\in U(n(\xi)).\\
\end{array}\right.
\end{equation}
where  $M^{r}(k)$ solves the pure soliton  RH problem  by ignoring the jump conditions of RHP \ref{RHP5},
  $M^{lo}(k)$  is a localized model to  match  parabolic cylinder  functions  in a neighborhood of each stationary  point $k_j$,
 and $E(k)$ is an error function computed by using a small-norm RH problem.

\subsubsection{Contribution from discrete spectrum }

\qquad We first deal with the function $M^{r}(k)$, which satisfying  the  RH problem as follows
\begin{RHP}\label{rhp-mr}
Find a matrix-valued function  $M^{r}(k)$ with properties:
\begin{itemize}
  \item $M^{r}(k)$ is analytical  in $\mathbb{C}\setminus \left\lbrace\zeta_n,\bar{\zeta}_n \right\rbrace, n=1,\cdots,4N_1+2N_2$;
  \item $
M^{r}(k) = I+\mathcal{O}(k^{-1}),\quad k \rightarrow \infty;$
\item $M^{r}(k)$ has   the same kind of residue conditions with  $M^{R}(k)$,

\end{itemize}
\end{RHP}
which can be shown to admits a unique solution.

\subsubsection{Contribution from contour near phase points}
\qquad In this part, we analyze the function $M^{lo}(k)$. Denote
$$
\Sigma^{lo}_j=\Sigma^{(2)}\cap U_{k_j}, \ \ \Sigma^{lo}= \Sigma^{(2)}\cap U(n(\xi)).
$$

Now we consider the following RH problem
\begin{RHP}\label{lorhp}															Find a matrix-valued function  $M^{lo}(k)$ satisfying the following properties:
\begin{itemize}
  \item Analyticity: $M^{lo}(k)$ is analytical  in $\mathbb{C}\setminus \Sigma^{lo}$;
  \item Jump condition: $M^{lo}(k)$ has continuous boundary values $M^{lo}_\pm(k)$ on $\Sigma^{lo}$ and																\begin{equation}															M^{lo}_+(k)=M^{lo}_-(k)V^{(2)}(k),\hspace{0.5cm}k \in \Sigma^{lo};														\end{equation}
  \item Asymptotic behaviors:																\begin{align}																		M^{lo}(k) =& I+\mathcal{O}(k^{-1}),\hspace{0.5cm}k \rightarrow \infty.															\end{align}
\end{itemize}
			
\end{RHP}

 The solution $M^{lo}(k)$ of this local  RH problem  can be constructed as follows \cite{HG2009}:
	\begin{align}
	M^{lo}(k)=I+t^{-1/2}\sum_{ j=1 }^{n(\xi)}\frac{A_j(\xi)}{k-k_j} +\mathcal{O}(t^{-1}),
	\end{align}
	where
	\begin{align}
	A_j(\xi)=\frac{i\eta}{2}\left(\begin{array}{cc}
	0 & [M^{pc}_1(k_j)]_{12}\\
	-[M^{pc}_1(k_j)]_{21} & 0
	\end{array}\right),
	\end{align}
 and  $M^{pc}(\zeta)$ is a solution of
  a solvable PC model RH problem with the asymptotics
 	\begin{align}
 	M^{pc}(\zeta  ) =& I+M^{pc}_1\zeta^{-1}+\mathcal{O}(\zeta^{-2}),\hspace{0.5cm}\zeta \rightarrow \infty.
 	\end{align}

\subsubsection{Error estimate}
\qquad In this subsection, we discuss the error function $E(k)$, which satisfies the  RH problem
\begin{RHP}
    Find a matrix-valued function $E(k)$  admitting the following properties:
\begin{itemize}
  \item Analyticity: $E(k)$ is analytical  in $\mathbb{C}\setminus  \Sigma^{ E } $, where
$$\Sigma^{ E }= \partial U(n(\xi))\cup
(\Sigma^{(2)}\setminus U(n(\xi));$$
  \item Asymptotic behaviors:
\begin{align}
&E(k ) \sim I+\mathcal{O}(k^{-1}),\hspace{0.5cm}k \rightarrow \infty;
\end{align}
  \item Jump condition: $E(k )$ has continuous boundary values $E_\pm(k )$ on $\Sigma^{ E }$ satisfying
$$E_+(k )=E_-(k )V^{E}(k),$$
where
\begin{equation}
V^{ E }(k)=\left\{\begin{array}{llll}
M^{ r }(k)V^{(2)}(k)M^{ r }(k)^{-1}, & k\in \Sigma^{(2)}\setminus U(n(\xi),\\[4pt]
M^{ r  }(k)M^{lo}(k)M^{ r  }(k)^{-1},  & k\in \partial U(n(\xi)).
\end{array}\right.
\end{equation}
\end{itemize}
\end{RHP}

The jump matrix $V^{ E }(k)$ has the following estimates
\begin{align}
&\parallel V^{ E }(k)-I\parallel_{L^p}\lesssim
e^{-ct  },   k\in \Sigma^{(2)}\setminus U(n(\xi)),\nonumber\\[4pt]
& | V^{ E }(k)-I|=   \big|M^{ r }(k)^{-1}(M^{lo}(k)-I)M^{ r }(k) \big| = \mathcal{O}(t^{-1/2}), \ k\in  \partial U(n(\xi)).\nonumber
\end{align}

The   existence and uniqueness  of $E(k)$ can  be proved  by using  a  small-norm RH theorem  \cite{DZ1,DZ2}.
 Moreover,   the  estimates  of $E(k)$  can be given as follows:
	\begin{align} E(k;\xi)=E_0+E_1(k-i)+\mathcal{O}((k-i)^2), \ k\to i,
	\end{align}
	where
	\begin{align}
&E_0=I+t^{-1/2}H_0+\mathcal{O}(t^{-1}),\\
&	E_1=t^{-1/2}H_1+\mathcal{O}(t^{-1}),\label{E1t}
	\end{align}
with
	\begin{align}
&H_0=\sum_{ j=1 }^{n(\xi)}\frac{1}{k_j-i} M^{r}(k_j)^{-1}A_j(\xi)M^{r}(k_j).\\
&	H_1=-\sum_{ j=1 }^{n(\xi)}\frac{1}{(k_j-i)^2} M^{r}(k_j)^{-1}A_j(\xi)M^{r}(k_j).
	\end{align}

\subsection{Contribution from  a pure $\bar{\partial}$-problem}

\qquad Now we consider the asymptotic behavior of $M^{(3)}(k)$ defined by \eqref{trans4-1}.
The  solution of the $\bar{\partial}$-problem for  $M^{(3)}(k)$ is equivalent to the integral equation
\begin{equation}
M^{(3)}(k)=I+\frac{1}{\pi}\iint_\mathbb{C}\dfrac{M^{(3)}(s)W^{(3)} (s)}{s-k}dm(s),\nonumber
\end{equation}
 which can be shown to  exist    a unique solution   and
have the following estimate

	\begin{equation}
		M^{(3)}(k)= M^{(3)}_0   +\mathcal{O}((k-i) ), \ k\to i,
	\end{equation}
   where
	\begin{align}
	   &	  M^{(3)}_0= I + \mathcal{O}( t^{-1+\kappa}),  \ \ 0<\kappa<1/2.
	\end{align}

\subsection{Proof of Theorem 1--Case II. }

\qquad Now we   complete the proof of Theorem 1 for the Case II.
Inverting the transformations (\ref{trans1}), (\ref{trans2-1}), (\ref{trans3-1}), (\ref{trans4-1}) and (\ref{trans5-1}), we have
\begin{align}
M(k)=&F(k)M^{(3)}(k)E(k)M^{r}(k)R^{(2)}(k)^{-1}T(k)^{-\sigma_3}. \label{23234}
\end{align}
Take    $k\to i$  along the path out of $ \Omega $, then $R^{(2)}(k)=I$.

Making  asymptotic expansions to  all
functions in (\ref{23234}) at $k= i$, we obtain that
\begin{align}		
M(k)=M_0+M_1(k-i)+\mathcal{O}((k-i)^2),		
	\end{align}		
	where
\begin{align}
&M_0(y,t)=M_0^0(y,t)+M_0^1t^{-1/2}+\mathcal{O}(t^{-1+\kappa}),\\ &M_1(y,t)=M_1^0(y,t)+M_1^1t^{-1/2}+\mathcal{O}(t^{-1+\kappa})\\
\end{align}
with
\begin{align}
&M_0^0(y,t)=F_1T_0^{-\sigma_3},\quad
M_0^1(y,t)=F_1H_0T_0^{-\sigma_3},\nonumber\\
&M_1^0(y,t)=F_2T_0^{-\sigma_3}+F_1T_1^{-\sigma_3},\quad
M_1^1(y,t)=(F_1H_1+F_2H_0)T_0^{-\sigma_3}+F_1H_0T_1^{-\sigma_3}.\nonumber
\end{align}
\qquad  Via  the reconstruction formula \eqref{restr},
we get the long-time asymptotic behavior for the solutions   of the 2-mCH equation \eqref{2mch-p}-\eqref{2mch-q}   as follows

\begin{align}
&( \log p )_x  =   p^{sol}(x,t) +g_{1}t^{-1/2}+\mathcal{O}(t^{-1+\kappa}),\nonumber\\
&( \log q )_x  =  q^{sol}(x,t) +g_{2}t^{-1/2}+\mathcal{O}(t^{-1+\kappa}),\nonumber
\end{align}
with
\begin{equation}
x(y,t)=y +S_2(y,t)+\mathcal{O}(t^{-1+\kappa}),
\end{equation}
where
\begin{align}
&p^{sol}(x,t)=-\frac{G_{31,x}}{G_{11}G_{31}}-1,\quad q^{sol}(x,t)=1-\frac{G_{31,x}}{G_{21}G_{31}},\label{pqr2}\\
&g_1(x,t)=\frac{G_{12}G_{31,x}}{G_{11}^2G_{31}}-\frac{G_{31}G_{32,x}-G_{31,x}G_{32}}{G_{11}G_{31}^2},\label{g1}\\
&g_2(x,t)=\frac{G_{22}G_{31,x}}{G_{21}^2G_{31}}-\frac{G_{31}G_{32,x}-G_{31,x}G_{32}}{G_{21}G_{31}^2},\label{g2}\\
&S_2(y,t)=\log\left(\frac{M_{0,11}^0}{M_{0,22}^0}+\frac{M_{0,22}^0M_{0,11}^1-M_{0,11}^0M_{0,22}^1}{(M_{0,22}^0)^2}t^{-1/2}\right)+\log(q-q_x)-\log(p+p_x),\label{s2}
\end{align}
and
\begin{align*}
&G_{11}:=\frac{M_{0,21}^0M_{0,22}^0}{M_{0,11}^0M_{0,12}^0}-1,\
G_{21}:=\frac{M_{0,11}^0M_{0,12}^0}{M_{0,21}^0M_{0,22}^0}-1,\
G_{31}:=M_{1,12}^0M_{1,21}^0,\\[4pt]
&G_{12}:=\frac{(M_{0,21}^0M_{0,22}^1+M_{0,22}^0M_{0,21}^1)M_{0,11}^0M_{0,12}^0-(M_{0,11}^0M_{0,12}^1+M_{0,12}^0M_{0,11}^1)M_{0,21}^0M_{0,22}^0}{(M_{0,11}^0)^2(M_{0,12}^0)^2},\\[4pt]
&G_{22}:=\frac{(M_{0,11}^0M_{0,12}^1+M_{0,12}^0M_{0,11}^1)M_{0,21}^0M_{0,22}^0-(M_{0,21}^0M_{0,22}^1+M_{0,22}^0M_{0,21}^1)M_{0,11}^0M_{0,12}^0}{(M_{0,21}^0)^2(M_{0,22}^0)^2},\\
&G_{32}:= M_{1,12}^0M_{1,21}^1+M_{1,21}^0M_{1,12}^1.
\end{align*}
Then   Theorem \ref{last} is proved.
\vspace{3mm}	
	
	\noindent\textbf{Acknowledgements}

	This work is supported by  the National Natural Science
	Foundation of China (Grant No. 12271104,  51879045).\vspace{2mm}
	
	\noindent\textbf{Data Availability Statements}
	
	The data that supports the findings of this study are available within the article.\vspace{2mm}
	
	\noindent{\bf Conflict of Interest}
	
	The authors have no conflicts to disclose.

																\end{document}